\documentclass[11]{article}

\usepackage{amsmath,amssymb}
\usepackage{graphicx}
\usepackage{subfigure,tabularx,multirow}
\usepackage{color}
\usepackage{url}
\usepackage{wrapfig}
\usepackage[ruled,vlined]{algorithm2e}
\usepackage[backref=page,unicode,bookmarks=false]{hyperref}

\newtheorem{theorem}{Theorem}
\newtheorem{problem}{Problem}
\newtheorem{lemma}{Lemma}
\newtheorem{proposition}{Proposition}

\newcommand{\df}[1]{\mbox{\textbf{#1}}}
\def\Cost{\text{\emph{routing cost}}}
\def\aCost{\text{\emph{additional cost}}}
\newcommand{\gv}{\widetilde{G}}
\newcommand{\ink}{I}
\newcommand{\cc}{C}
\newcommand{\inkcoeff}{k_{\textit{ink}}}
\newcommand{\lencoeff}{k_{\textit{len}}}
\newcommand{\capcoeff}{k_{\textit{cap}}}

\newenvironment{proof}{\begingroup\Proof}{\qed\endgroup}
\def\Proof{\noindent{\bf Proof\/:}\nobreak}
\def\qed{\unskip~{\vrule height 4pt depth 0pt width 4pt}\medbreak}

\title{Edge Routing with Ordered Bundles}

\author{Sergey Bereg\thanks{
University of Texas at Dallas, USA
E-mail: {\tt besp@utdallas.edu}}
\and
Alexander E. Holroyd\thanks{
Microsoft Research, USA
Email: {\tt holroyd@microsoft.com}}
\and
Lev Nachmanson\thanks{
Microsoft Research, USA
Email: {\tt levnach@microsoft.com}}
\and
Sergey Pupyrev\thanks{
Ural Federal University, Russia
Email: {\tt spupyrev@gmail.com}}
}

\date{}

\begin{document}
\maketitle

\begin{abstract}
Edge bundling reduces the visual clutter in a drawing
of a graph by uniting the edges into bundles. We propose a method of edge
bundling drawing each edge of a bundle separately as in metro-maps and call our method ordered
 bundles. To produce aesthetically looking edge routes it minimizes a cost function on the edges.
The cost function depends on the ink,  required to draw the edges, the edge lengths, widths and separations.
The cost also penalizes for too many edges passing through narrow channels by using the constrained Delaunay triangulation.
The method avoids unnecessary edge-node and edge-edge crossings.
To draw edges with the minimal number of crossings and separately within the same bundle
we develop an efficient algorithm solving a variant of the metro-line crossing minimization problem.
In general, the method creates clear and smooth edge routes giving
an overview of the global graph structure, while still drawing each edge separately and thus enabling local analysis.
\end{abstract}




\section{Introduction}
\label{sec:intro}
Graph drawing is an important visualization tool for exploring and
analyzing network data.
The core components of most graph drawing algorithms are
computation of positions of the nodes and edge routing. In this paper
we concentrate on the latter problem.

For many real-world graphs with substantial numbers of edges
traditional algorithms produce visually cluttered layouts. The
relations between the nodes are difficult to analyze by looking
at such layouts. Recently, edge bundling techniques have been
developed in which some edge segments running close to each other are
collapsed into bundles to reduce the clutter.  While these
methods create an overview drawing, they typically allow the
edges within a bundle to cross and overlap each other
arbitrarily, making individual edges hard to follow.  In
addition, previous approaches allow the edges to overlap the nodes and obscure their text or graphics.

We present a novel edge routing algorithm for undirected
graphs. The algorithm
produces a drawing in a ``metro line'' style (see Fig.~\ref{fig:jazz}).  The graphs for
which our algorithm is best applicable are of medium size with
a large number of edges, although it can process larger graphs
efficiently too.

\begin{figure}[t]
	\centering
	\subfigure[]{
	\includegraphics[height=5.0cm]{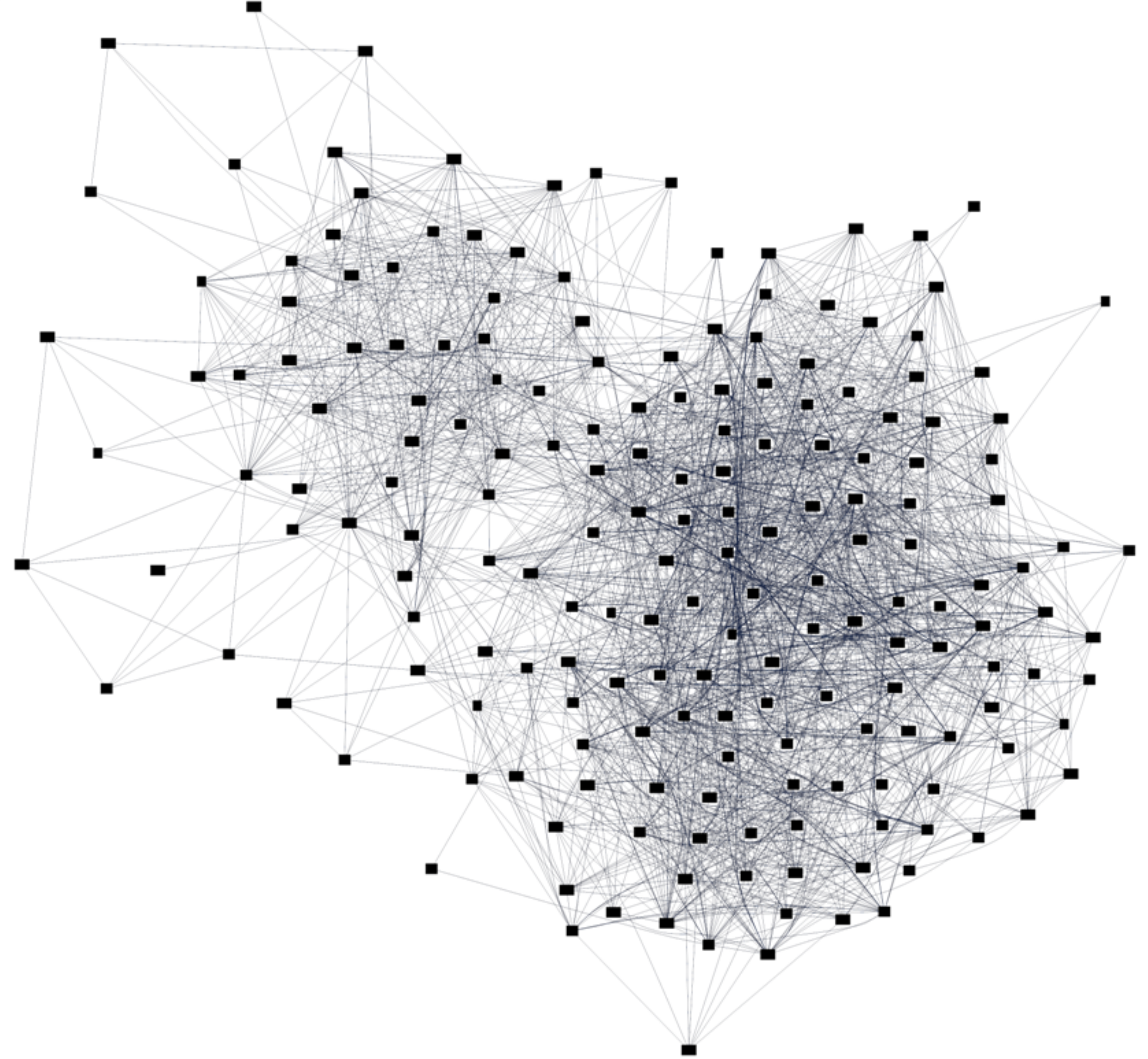}}
~~~~~~~~~~~~~
	\subfigure[]{
	\includegraphics[height=5.0cm]{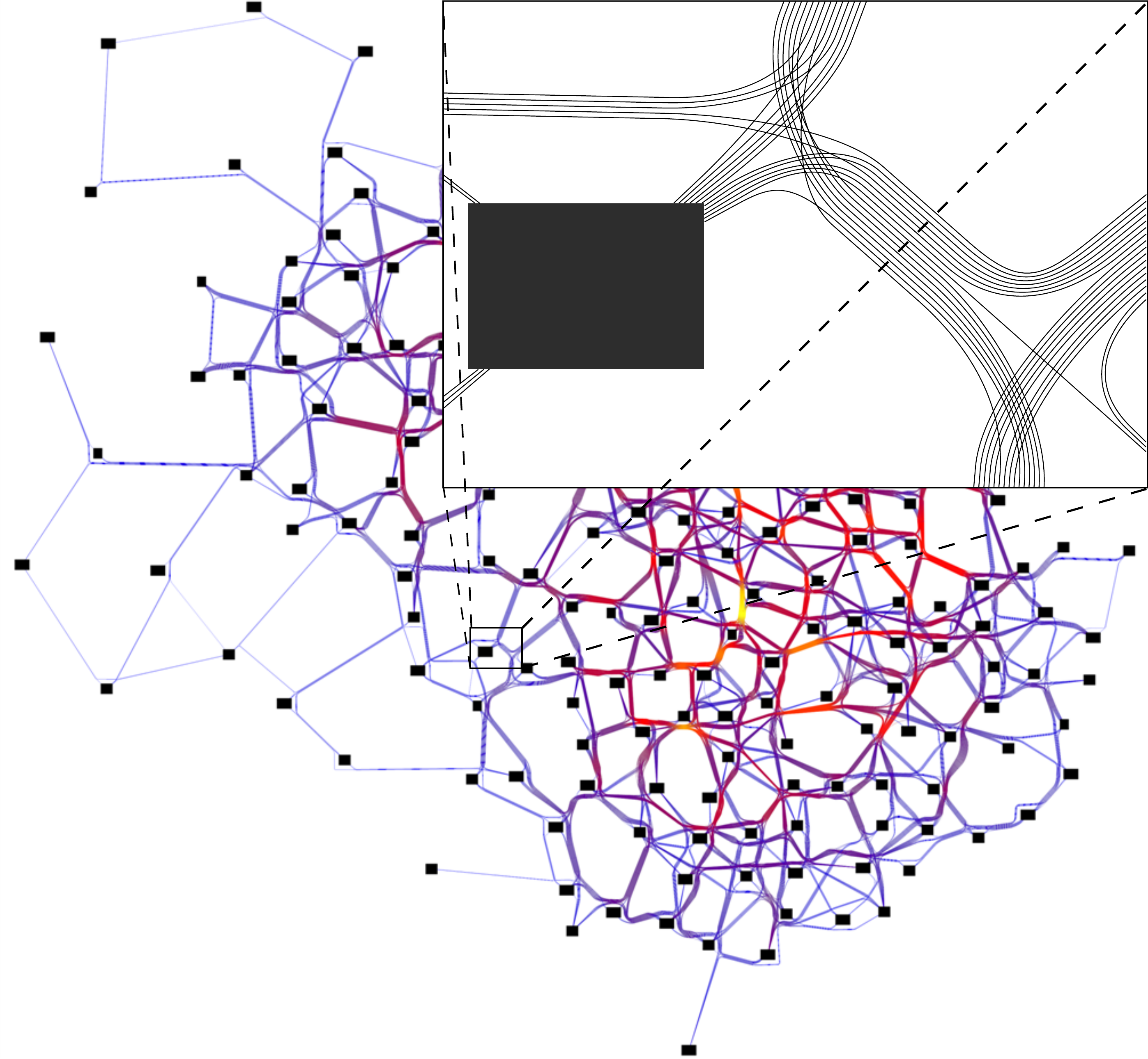}}
	\caption{Edge bundling example (\texttt{jazz} graph)}
    \label{fig:jazz}
\end{figure}

The input for our algorithm is an undirected graph with given
node positions. These positions can be generated by a graph
layout algorithm, or, in some applications (for instance,
geographical ones) they are fixed in advance. During the
algorithm the node positions are not changed. The main steps of
our algorithm are similar to existing approaches, but with
several innovations, which we indicate with \emph{italic text} in the
following description.

\textbf{Edge routing.} The edges are routed along the edges of some underlying graph
and the overlapping parts are organized into bundles.
One approach here has been to minimize the total ``ink'' of
the bundles. However, this often produces excessively long paths.
\emph{For this reason we introduce a novel cost function for
edge routing}. Minimizing this function forces the paths to share routes,
creating bundles, but at the same time it keeps the paths relatively
short. Additionally, the cost function penalizes routing too
many edges through narrow gaps between the nodes. \emph{We furthermore show how
to route the bundles around the nodes.}

\textbf{Path construction.} In this step the paths belonging to the
same bundle are ``nudged'' away from each other. The effect of
this action is that individual edges become visible and the
bundles acquire thickness. \emph{Contrary to previous
approaches, we try to draw the edges of a bundle as parallel as
possible with a given gap.} However, such a routing might not
always exist because of the limited space between the nodes.
 We provide a heuristic that finds a drawing with bundles of a suitable thickness.

\textbf{Path ordering.} To route individual edges, an order of
the edge segments within each bundle needs to be computed.
This order minimizes the number of crossings between the edges of
the same bundle. The problem of finding such an order is
a variant of the metro-line crossing minimization problem.
\emph{We provide a new efficient algorithm that solves this
problem exactly}.

The next section summarizes related work. In
Section~\ref{sec:algorithm} we introduce the ordered edge bundling
algorithm, followed by the detailed explanation of the algorithm steps
 in Sections~\ref{sec:edge-routing}-\ref{sec:order}.
Results of experiments are presented in
Section~\ref{sec:experiments}. We discuss some additional issues and future work in
Section~\ref{sec:conclusion}.


\section{Related work}
\textbf{Edge routing.}
The problem of drawing edges as curves that avoid node interiors
has been considered in several graph drawing papers. Dobkin~et~al.~\cite{dobkin97}
initiated a visibility-based approach in which the shortest
obstacle-avoiding route for an edge is chosen. The method requires the computation of the visibility graph, which
can be too slow for graphs with hundreds of nodes. An improvement
was introduced by Nachmanson~et al.~\cite{lev09}, achieving faster
edge routing using sparse visibility graphs. The latter method
constructs approximate shortest paths for the edges.  We use
the method as the starting point of our algorithm.

Dwyer~et~al.~\cite{dwyer07} proposed a force-directed method, which is able to
improve a given edge routing by straightening and shortening
edges that do not overlap node labels. However, a feasible
initial routing is needed. Moreover, the method moves the nodes, which is
prohibited in our scenario.

The problem of finding obstacle-avoiding paths in the plane arises in several
other contexts. In robotics, the motion-planning problem is to find a
collision-free path among polygonal obstacles. This problem is usually solved
with a visibility graph, a Voronoi diagram, or a combination thereof~\cite{wein07}.
Though these methods can be used to route a single edge in a graph, they do
not apply directly to edge bundling. Many path routing problems were studied
in very-large-scale integration (VLSI) community. We use some ideas of~\cite{cole84,dai91} as will be explained later.

\textbf{Edge bundling.}
The idea of improving the quality of layouts via edge bundling is
related to edge concentration~\cite{newbery89}, in which a layered graph is covered by bicliques in order to reduce the number of
edges. Eppstein~et~al.~\cite{eppstein03,eppstein07} proposed
a representation of a graph called a confluent drawing. In a confluent
drawing a non-planar graph is represented by
merging crossing edges to one, in the manner of railway tracks.

We believe that the first discussion of bundled edges in the graph
drawing literature appeared in~\cite{koren06}.  There the authors
improve circular layouts by routing edges either on the outer
or on the inner face of a circle. In that paper the edges are
bundled by an algorithm that tries to minimize the total
ink of the drawing.  Here we follow a similar strategy, but in
addition we try to keep the edges themselves short (among other criteria).

In the hierarchical approach of \cite{holten06} the edges are
bundled together based on an additional tree structure.
Unfortunately, not every graph comes with a suitable underlying
tree, so the method does not extend directly to general
layouts.  In~\cite{sergey10} edge bundles were computed for
layered graphs.  In contrast, our method applies to general
graph layouts.

Edge bundling methods for general graphs are given
in~\cite{cui08,hu11,holten09,lambert10}. In the force-directed
approach~\cite{holten09} the edges can attract each other to
organize themselves into bundles. The method is not efficient,
and while it produces visually appealing drawings, they are
often ambiguous in a sense that it is hard to follow an individual edge.
The approaches~\cite{cui08,lambert10} have a common feature:
they both create a grid graph for edge routing. Our method also
uses a special graph for edge routing, but our approach is
different;  we modify this graph to obtain better edge
bundles. And unlike~\cite{cui08,hu11,holten09}, we avoid edge-node
overlaps.

Recently there were several attempts to enhance edge bundled
graph visualizations via image-based rendering techniques.
Colors and opacity were used to highlight the density of
overlapping lines~\cite{cui08,holten06}.  Shading, luminance,
and saturation encode edge types and edge
similarity~\cite{ersoy11}. We focus on the geometry of the edge bundles only.
The existing rendering techniques may be applied on top of our routing scheme.

\textbf{Path ordering.}
The problem of ordering paths along the edges of an embedded graph
is known as metro-line crossing minimization (MLCM) problem. The problem was introduced by
Benkert~et~al.~\cite{benkert07}, and was studied in several
variants in~\cite{asquith08,bekos08,argyriou09,nollenburg09}.
We mention two variants of the MLCM problem that are related
to our path ordering problem.

Asquith~et~al.~\cite{asquith08} defined the so-called MLCM-FixedSE problem,
in which it is specified whether a path terminates at the top or
bottom of its terminal node.
For a graph $G=(V,E)$ and a set of paths $P$ they give the
algorithm with $O(|E|^{5/2}|P|^3)$ time complexity.
A closely related problem called MLCM-T1, in which all paths
connect degree-$1$ terminal nodes in $G$, was considered by
Argyriou~et~al.~\cite{argyriou09}.  Their algorithm computes
an optimal ordering of paths and has a running time of
$O((|E|+|P|^2)|E|)$. Recently, N{\"o}llenburg~\cite{nollenburg09}
presented an $O(|V||P|^2)$-time algorithm for these variants of
the MLCM problem.  Our algorithm can be used to solve both the
MLCM-FixedSE and MLCM-T1 problems with complexity $O(|V||P|+|V|\log |V|+|E|)$.
To the best of our knowledge, this is the fastest solution to date of
these variants of MLCM.

Crossing minimization problems have also been researched in the circuit
design community~\cite{groeneveld89a,mareksadowska95,chen98}.
Given a global routing of nets (physical wires) among the circuit modules, the
goal is to adjust the routing so as to minimize crossings between the pairs of the nets.
In this context, a problem similar to the path ordering problem was studied by
Groeneveld~\cite{groeneveld89a}, who suggested $O(|V|^2|P|)$ algorithm.
Another method, which works for graphs of maximum degree $4$, is given in~\cite{chen98}.
Marek-Sadowska~et~al.~\cite{mareksadowska95} considered the problem
of distributing the set of crossings among circuit regions. They
introduced constraints for the number of allowed crossings in each
region, and provided an algorithm to build an ordering with these
constraints. Their algorithm has $O(|V|^2|P|+|V|\xi^{3/2})$ complexity,
where $\xi$ is the number of path crossings.


\section{Algorithm}
\label{sec:algorithm}

In this section we introduce some terminology and give an overview of our algorithm.
The input to the edge bundling algorithm
comprises an undirected graph $G=(V,E)$, where $V$ is the set of nodes and $E$
is the set of edges. Each node of the graph comes with its \df{boundary curve}, which
is a simple closed curve. The boundary curves of different nodes are disjoint, and the
interiors they surround are disjoint too.
The output of the algorithm is a set of paths $P=\{\pi_{st}: st \in E\}$, where for an edge $st \in E$
its path is referenced as $\pi_{st}$; it is a simple curve in the plane connecting the
boundary curve of $s$ with the boundary curve of $t$. The input also includes a required
path width for each edge, and a desired path separation.

\begin{figure}[t]
	\centering
	\includegraphics[width=\linewidth]{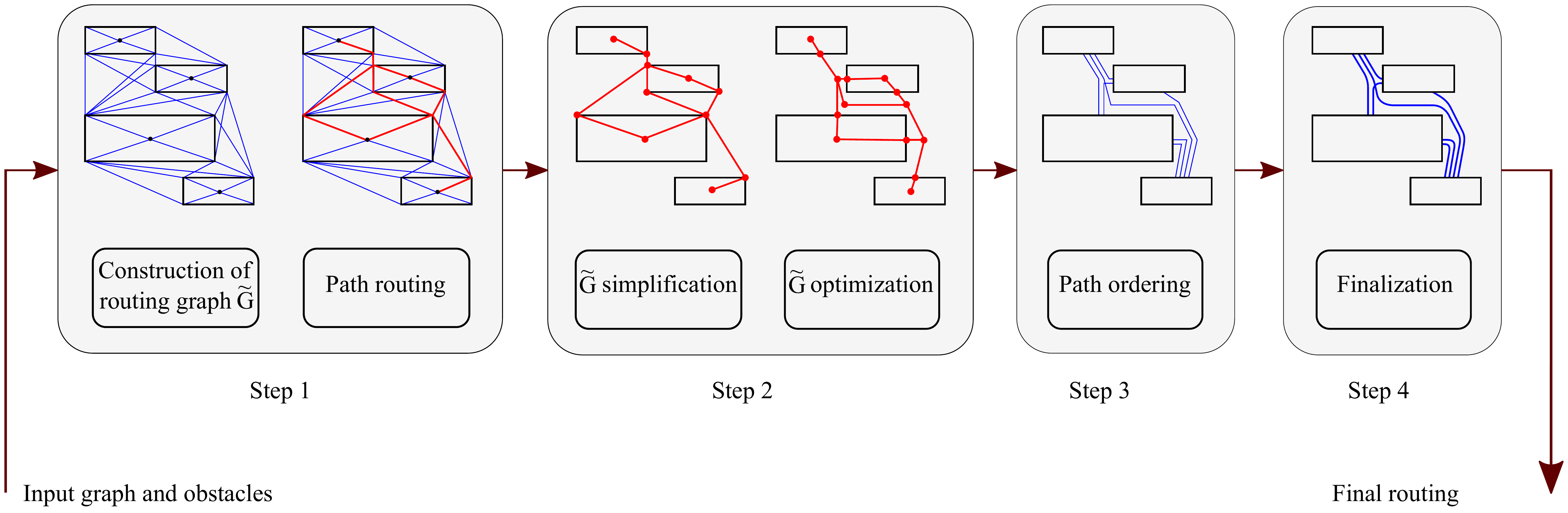}
	\centering
	\caption{Algorithm pipeline.}
    \label{fig:steps}
\end{figure}

\emph{Step 1} of the edge bundling algorithm produces an initial routing of the paths. This routing is done
on an auxiliary graph $\gv$ called the \df{routing graph}. The edges of $\gv$ are straight-line segments on the plane.

\emph{Step 2} improves the routing. Firstly, $\gv$ is cleaned up by removal of
all its nodes and edges that are not part of any path. Secondly, $\gv$ is modified to 
accommodate the paths together with the path widths and separations, and to improve the paths.

In steps 1 and 2 we use a notion of \df{ink} similar to \cite{koren06,sergey10}.
Given a graph embedded in the plane with straight-line edges,
the ink $\ink$ of a set of paths in the graph is the sum of the lengths of the
edges of the graph used in these paths; if an edge is used by several paths its
length is counted only once.  A set of paths sharing an edge of the
graph is called a \df{bundle}.  It is important to note that in the final drawing
 the paths in a bundle will be drawn as separate curves, so $\ink$ does not represent the amount of ink
 used in the final drawing.

\emph{Step 3}.  Consider a bundle with the shared edge $uv$ in $\gv$.
We would like to draw the segments of the paths that run along $uv$ as distinct curves parallel to $uv$.
This requires ordering of the paths of each bundle. Step 3 produces such an order for all bundles
 in a way that minimizes crossings (which occur at the nodes).

\emph{Step 4} draws each path as a smooth curve using straight-line,
cubic Bezier, or arc segments.

The overall algorithm is illustrated in Fig.~\ref{fig:steps}, and discussed in
detail in the next sections.

\section{Edge Routing}
\label{sec:edge-routing}

\subsection{Generation of the Routing Graph}
\label{sec:gengv}
As a preprocessing step, for each node $s \in V$ we create a convex polygon $O_s$
containing the boundary curve of $s$ in its interior. We keep the number of corners in $O_s$
bounded from above by some constant. We call $O_s$ an \df{obstacle},
and build the obstacles in such a way that $O_s$ and $O_t$ are disjoint for different $s$ and $t$.
For each $s \in V$ we choose a point inside of the boundary curve of $s$, which of course is
inside of $O_s$ too, and call this point a center of $O_s$, see Fig.~\ref{fig:obstacle}.

The algorithm starts with the construction of the routing graph $\gv$.
We follow the approach of~\cite{lev09} building a sparse visibility graph.
The vertices of $\gv$ comprise the obstacle centers, the obstacle corners,
 and some additional vertices  from the boundaries of $O_s$. These additional point are added to make $\gv$ rich enough for path generation. The points added to $O_s$ have the property that each cone with the apex at the center of $O_s$ and
with an angle given in advance, $\pi /6$ in our settings, contains at least one corner of $O_s$, or at least one new added point, see Fig~\ref{fig:gvnotshrunk}. To build the visibility edges of $\gv$ adjacent to vertex $u$, the vertex is surrounded by a family of cones covering the plane with the apexes at $u$, and in each cone at most one edge is created by connecting $u$ with another vertex of $\gv$ which is visible from $u$, belongs to the cone, and is closest to $u$ in the cone. The angles of the cones are  $\pi /6$ in our default settings.
\begin{figure}[t]
	\centering
	\subfigure[]{
	\includegraphics[height=1.5cm]{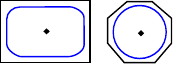}
    \label{fig:obstacle}}
~~~~~~~~~~~~~~~~~~~~~~
	\subfigure[]{
	\includegraphics[height=1.5cm]{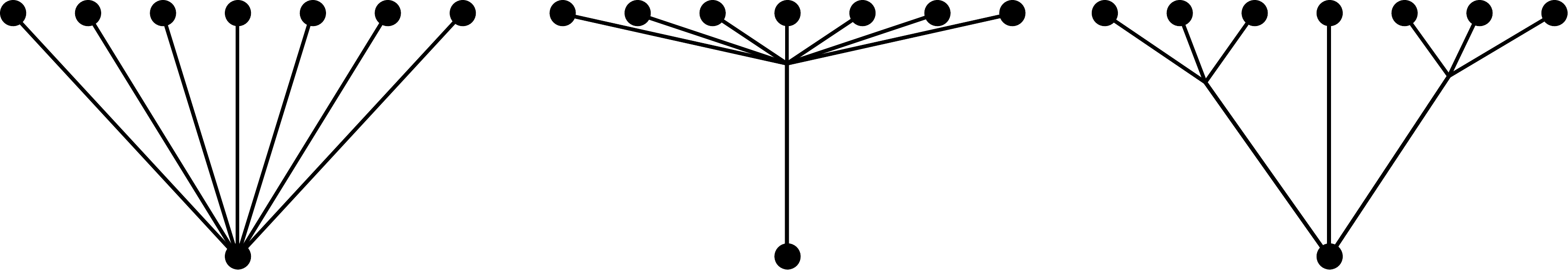}
    \label{fig:pureink}}
    \caption{(a)~Boundary curves (blue) of two nodes with its obstacles (black) and centers.
    (b)~(left)~a graph needed to draw; (center)~edge routing with small ink; 
    (right)~edge routing with larger ink and shorter paths.}
\end{figure}

\begin{figure}[t]
	\centering
	\subfigure[]{
	\includegraphics[height=3.5cm]{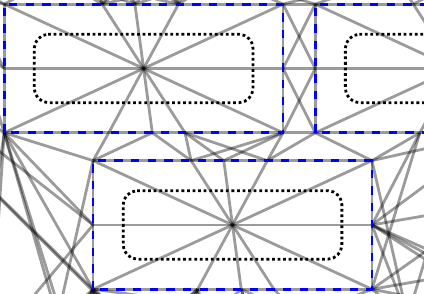}
    \label{fig:gvnotshrunk}}
~~~~~~~~~~~~~~~~~~~~~~
	\subfigure[]{
	\includegraphics[height=3.5cm]{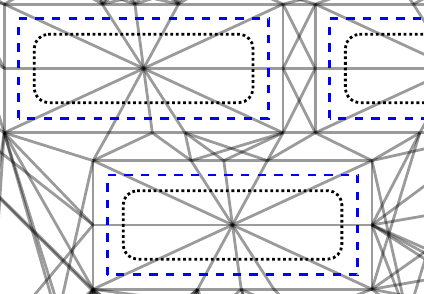}
    \label{fig:gvshrunk}}
    \caption{A fragment of $\gv$ is shown. (a)~The boundary curves are rendered with dotted curves, the obstacles with dashed curves, and
the edges of $\gv$ with solid curves.
    (b)~The obstacles are shrunk, so the visibility edges of $\gv$ do not cross or touch the obstacles.}
\end{figure}

The method works in $O(r\log r)$ time and creates $\gv$ with
$O(r)$ edges, where $r$ the number of vertices in $\gv$.

There are two types of edges in $\gv$: the visibility edges, and the edges from the obstacle centers.
A visibility edge has both its ends belonging to the obstacle boundaries and it does not intersect
the interior of any obstacle.  An edge of the second type is adjacent to an obstacle center.
  Such an edge does not intersect the interior of any obstacle except of the one
containing the center, see Fig.~\ref{fig:gvnotshrunk}. We denote the set of nodes of $\gv$ by $W$, and
the set of edges of $\gv$ by $U$. By construction $r=O(|V|)$, and therefore $\gv$ has $O(|V|)$ nodes and $O(|V|)$ edges.

Further on, it is necessary for our method to keep the visibility edges of $\gv$ disjoint from the obstacles.
For this purpose after computing $\gv$ we shrink the obstacles slightly, so that each obstacle still contains its boundary curve,
but the visibility edges do not touch or intersect the obstacles anymore, see Fig.~\ref{fig:gvshrunk}.

Next we route the paths on $\gv$.

\subsection{Routing Criteria and Routing Cost}
To create bundles we apply the idea of ink minimization~\cite{koren06,sergey10,hu11}.
However, minimizing the ink alone does not always produce satisfactory results; some
paths in a routing with a small ink might have sharp angles or long
detours as shown in Fig.~\ref{fig:pureink}. We improve the paths by introducing
additional criteria.

\begin{itemize}
\item {\em Short Path Lengths}. The length of a path $\pi_{st}$ is denoted by $\ell_{st}$.
One approach would be to minimize the sum of $\ell_{st}$.  However, this gives insufficient weight to paths that connect relatively close pairs of nodes. For this reason
we instead consider the normalized lengths $\ell_{st}/|st|$, where $|st|$ is the
distance between the centers of $O_s$ and $O_t$.

\item {\em Capacity}.
We would like to avoid routing too many paths through a narrow gap.
For this purpose we adapt an idea of~\cite{cole84,dai91}, where capacities are associated with
segments connecting obstacles. To construct these segments, we compute the
constrained Delaunay triangulation (CDT)~\cite{DBLP:journals/gis/DomiterZ08} with the constrained edges
 being the obstacle sides. A Delaunay edge connecting two different obstacles is called a \df{capacity segment}.
We penalize for too many paths passing through a capacity segment.

Let us define \df{capacity} $c_\sigma$ of a capacity segment $\sigma$. Suppose $\sigma$ connects points $a \in A$ and $b \in B$
for obstacles $A$ and $B$. Then $c_\sigma = ( |a,B| + |b,A|)/2$,
where $|a,B|$ is the distance from $a$ to $B$, and $|b,A|$ is the distance from $b$ to $A$.

During routing a path we assign it to the capacity segments intersected by the path, see Fig.~\ref{fig:capacitySegs}.
\begin{figure}[t]
	\centering
	\includegraphics[height=3cm]{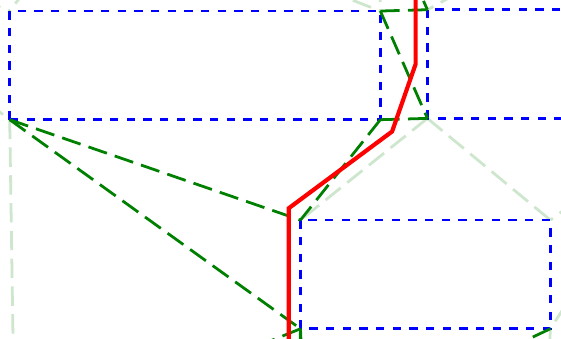}
  
\caption{ The solid curve represents a path, the long dashed opaque line segments represent the capacity segments crossed by the path, short dashed line segments represent the shrunk obstacle boundaries.}
  \label{fig:capacitySegs}
\end{figure}
For each capacity segment $\sigma$ we select all paths assigned to $\sigma$, and calculate
\df{routing width} $w_\sigma$ as the total path
widths plus the total path separation.
\df{Capacity overflow penalty}  of $\sigma$ is defined as
$p_\sigma=(w_\sigma-c_\sigma)^+$.
The \df{total capacity overflow} {\bf penalty} is defined as $C=\sum_\sigma p_\sigma$, where the sum is taken over all capacity segments.
\end{itemize}

Minimizing the ink, the paths lengths, and the capacity overflow penalty often is impossible at the same time.
For example, smaller ink usually leads to longer paths.
The capacity penalty can force some overlapping paths to go through different channels thus possibly increasing the ink.
To find a solution satisfying all the criteria we consider a multi-objective
optimization problem: minimize  the \df{routing cost} of a set of paths $P$ defined as

$$\Cost = \inkcoeff \ink + \lencoeff \sum_{\pi_{st}\in P} \frac{\ell_{st}}{|st|} + \capcoeff \cc .$$

We stress here that the $\Cost$ is used only in steps 1 and 2 of our algorithm. Steps 3 and 4
producing the final routing preserve the overflow penalty and try to preserve the path lengths.
However the notion of ink, as we defined it, is not applicable to the final routing.

\subsection{Routing of Paths}
\label{path routing}
We try to route the paths on $\gv$ with the minimal $\Cost$.
The problem is formulated as follows.

\begin{problem}[Path Routing]
\label{problem:MIER} Given the graph $\gv=(W,U)$ and a set of pairs
of nodes $(s_i t_i) \in W^2$, find paths between $s_i$ and
$t_i$ for all $i$ so that $\Cost$ is minimized.
\end{problem}

Here $s_i$ and $t_i$ correspond to the centers of obstacles connected by the
edges of $E$. Problem~\ref{problem:MIER} is NP-hard, because its instance
with $\inkcoeff = 1$, $\lencoeff = \capcoeff = 0$ is a Steiner
Forest Problem~\cite{garey79}. Therefore we suggest a heuristic to optimize
$\Cost$. In this heuristic we solve an easier task, where some paths are already known
and we need to route the next path. We will route it by minimizing an $\aCost$,
which is the increment of the $\Cost$ associated with this path. For a path $\pi_{st}$ we have
$\aCost = \inkcoeff\ \Delta\ink + \lencoeff\ \ell_{st}/|st| + \capcoeff \ \Delta C$.
Here $\Delta\ink$ is the increment in the ink, which is the sum of
edge lengths of $\pi_{st}$ that were not part of any previous path; $ \Delta C$ is the growth
of $\sum_\sigma {p_\sigma}$, where the sum is taken over all capacity segments $\sigma$ assigned
to path $\pi_{st}$.

\begin{problem}[Single Path Routing]
\label{problem:MIERS} Given the graph $\gv=(W,U)$ and a set $P$ of already routed paths
on $\gv$, find a path from $s$ to $t$ so that the $\aCost$ is minimized.
\end{problem}

To solve the problem, for each edge $e$ of $\gv$ we assign weight $ \inkcoeff\ \delta_e+\lencoeff\
\ell_e / |st| + \capcoeff \ \Delta C_e$ to $e$ . Here $\delta_e$ is $\ell_e$ if $e$ is not taken by a
previous path, and $0$ otherwise. The value of $\Delta C_e$ is the growth of $\sum_\sigma {p_\sigma}$
for the case the path is passing through $e$, where the sum is taken
over all capacity segments $\sigma$ crossed by $e$.

It can be seen that the minimum of $\aCost$ is achieved by a shortest path from $s$ to $t$
when we use these weights. We apply the Dijkstra's algorithm to find a path solving the problem.

To find a solution for Problem~\ref{problem:MIER},
we organize paths in a sequence $(s_1 t_1),\dots, (s_m t_m)$,
and iteratively solve Problem~\ref{problem:MIERS} for $(s_k t_k)$,
with the paths $(s_i t_i)$, $i<k$ already routed, for $k=1,\dots,m$.

 The routing of a single path takes $O(|U|\log |W|)$ time with the Dijkstra's algorithm.
 Routing all paths amounts to $O(|E||U|\log |W|)$ steps.

\begin{problem}[Multiple Path Routing]
\label{problem:MIERSS} Given the graph $\gv$ with some paths
already routed, find paths for $(s^* t_1),\dots,(s^* t_k)$ so
that $\aCost$ is minimized.
\end{problem}

We can solve this problem by a dynamic programming approach. We
first fix a set of pre-existing paths in $\gv$; $\aCost$ will
always be with respect to these paths. Let us call a
\df{state} a pair $(v,P)$, where $v$ is a node of $\gv$, and
$P$ is a subset of $\{t_1,\dots,t_k\}$. We need to solve our
problem for the state $(s^*,\{t_1,\dots,t_k\})$. We reduce the
problem to solving it for ``smaller'' states, that are the
states with fewer elements in $P$. For a state $(v, P)$ we
define its cost $f(v,P)$ as the minimal $\aCost$ of a set of
paths $\{(v, t), t \in P\}$.  A set of paths giving the minimal
$f(v,P)$ is called an \df{optimal} set for state $(v,P)$. Let
us clarify the structure of an optimal set of paths.

By the subgraph generated by a set of paths in $\gv$ we
mean the subgraph of $\gv$ comprising all edges and nodes in
the paths.

\begin{lemma}
\label{lemma:tree} For each state there exists an
optimal set of paths that generates a tree.
\end{lemma}

\begin{proof}
Let $\Pi$ be any optimal set of paths for state $(v,P)$, and
$G^{\Pi}$ be the graph generated by $\Pi$, and note that it is
connected.  Let $T$ be a shortest path tree of $G^{\Pi}$,
rooted at $v$, with respect to
ordinary edge lengths. Let $\Pi'$ be the set of paths
connecting $v$ to the points of $P$ in $T$. The $\aCost$ of
$\Pi'$ is at most that of $\Pi$. Indeed, the increment in
$\ink$ is no greater because $T$ is a subgraph of $G^{\Pi}$.
Each path of $\Pi'$ is shortest in $G^{\Pi}$ and thus no longer than
the corresponding path of $\Pi$. Hence, $\Pi'$ is an optimal
set for $(v,P)$.
\end{proof}

Lemma~\ref{lemma:tree} leads us to the following formula.

$$
f(v,P) = \min
\begin{cases}
f(u,P)+\inkcoeff\ \delta_{vu}+ \lencoeff\ \sum_{t \in P} \ell_{vu}/|s^*t|,
& \text{for $u\in W$ adjacent to $v$,}
\\
f(v,P')+f(v,P - P'), &\text{for $P'$ with } \emptyset\subset P'\subset P
\end{cases}
$$

The minimum is taken over both expressions on the right
as $u$ and $P'$ vary. To verify this, we consider some
optimal set of paths for $(v,P)$ that form a tree, and split
into two cases.  The first line corresponds to the case where
$u$ is the only neighbor of $v$ in the tree.  The second line
is the case where $v$ has at least two neighbors, thus the
paths can be partitioned into two proper subsets with no common
edges.

\begin{figure}[t]
	\subfigure[]{
	\includegraphics[height=3cm]{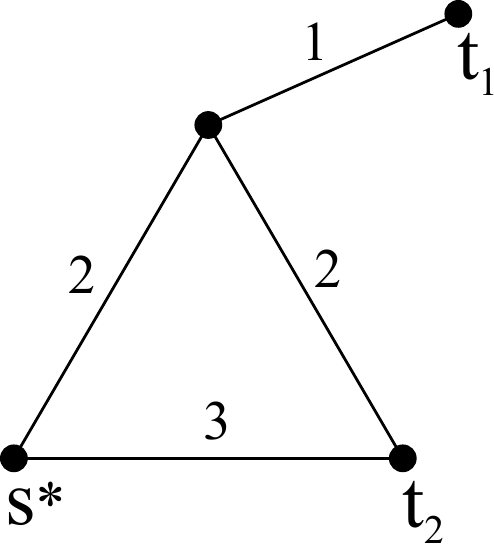}}
~~~~~~~~~~~~~~~~
	\subfigure[]{
	\includegraphics[height=3cm]{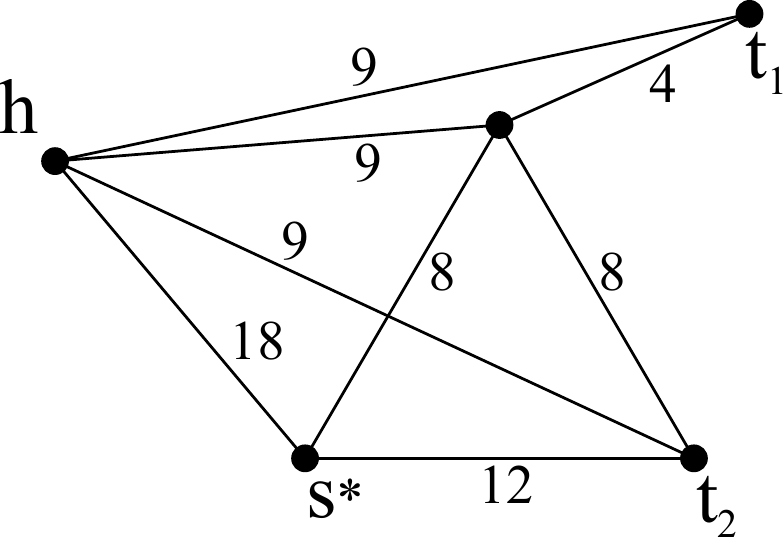}}
~~~~~~~~~~~~~~~~
	\subfigure[]{
	\includegraphics[height=3cm]{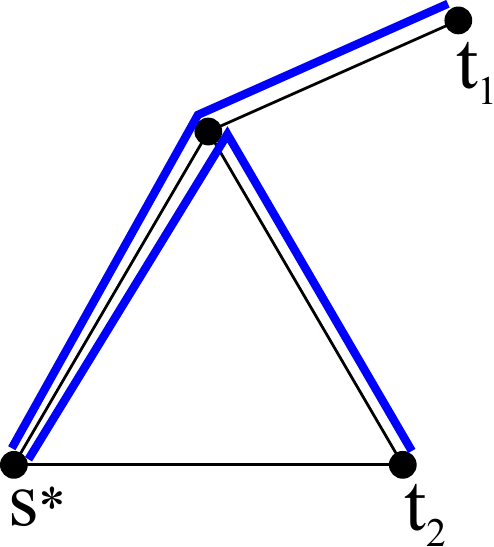}}
	\centering
	\caption{An example for Problem~\ref{problem:MIERSS} with
$\inkcoeff=2, \lencoeff=1, \capcoeff=0$, and $|s^*t_1|=|s^*t_2|=1$.
The goal is to find paths $s^*t_1$ and $s^*t_2$ with minimal $\aCost$.
(a)~Graph $\gv$, weights $\ell_e$ for each edge $e$ are shown.
(b)~Constructed graph $H$ with edge weights.
(c)~Optimal paths (blue) with $\aCost=17$.}
    \label{fig:mpr}
\end{figure}

Now we describe how to compute $f(v,P)$. Let us assume that
$f$ is known for all states $(u,P')$ for $P'$ a proper
subset of $P$. To compute $f(v,P)$, a new graph $H$ is
constructed with $\gv$ as a subgraph. An edge $e$ of $\gv$ has
weight $\inkcoeff\ \delta_e + \lencoeff\ \sum_{t \in P} \ell_{e}/|s^*t|$ in $H$.
 We add a new node $h$ to $H$ and connect it with all
nodes of $\gv$. For every new edge $hu$ we assign weight
$\min_{P'}{f(u,P')+f(u,P-P')}$, where $P'$ varies over proper
non-empty subsets of $P$ (see Fig.~\ref{fig:mpr}). One can see that the required value
$f(v,P)$ is the length of a shortest path from $v$ to $h$ in
graph $H$. We can compute it using Dijkstra's algorithm.

To solve Problem~\ref{problem:MIERSS} we work bottom-up. We
first compute all $f(v,P)$ with $|P|=1$ and $v$ is a node of
$\gv$, by the algorithm for Problem~\ref{problem:MIERS}, where
we find a path with the minimal $\aCost$.  Then we compute the
values $f(v,P)$ for each $v$ and $|P|=2, \ldots, k$ by creating
the corresponding graphs $H$. Finally, the answer for the
problem is $f(s^*,\{t_1,\dots,t_k\})$.

\paragraph{Running time}
The main steps of the above algorithm are the
construction of graph $H$, and finding a shortest path on it with the
Dijkstra's algorithm for each state $(v,P)$. Luckily, graph $H$ depends
only on the $P$ component of a state. The construction of graph $H$
for a fixed set $P$ takes $O(2^{|P|}|W|)$ time.  We execute the Dijkstra's
algorithm only once per $P$ starting from $h$ to compute $f(v,P)$ for
all $v\in W$. Thus, finding $f(v,P)$ for a known $H$ and for all $v\in
W$ takes $O(|W|\log |W|)$ time. Summing over all possible sets $P$ gives
$$
O\Bigl(\sum_{P}\bigl[2^{|P|}|W|+|W|\log |W|\bigr]\Bigr) = O\bigl(3^k|W| + 2^k|W|\log |W|\bigr).
$$

\section{Path Construction}
\label{sec:path construction}

\subsection{Structure of the paths}
Before proceeding to the optimization of $\gv$, we need to give more details on the structure of the paths that we produce and define some additional constructions on $\gv$.

We call a node $v \in W$ an \df{intermediate} node if it is not an obstacle
center. Recall that a bundle is a set of paths sharing the same edge of $U$.
That means that as soon as the paths are routed on $\gv$, the bundles become defined. In the final drawing we
would like to draw the paths of a bundle respecting the width and path
separation and outside of the obstacles. As described in Section~\ref{sec:gengv},
we shrink the obstacles after constructing $\gv$,
therefore each intermediate node of $W$ lies outside the obstacles; an obstacle intersects
an edge of $U$ only if the edge is adjacent to the obstacle center.
To give a general idea of the final path structure, we refer Fig.~\ref{fig:pathsdetailed}.

\begin{figure}[t]
    \centering
\subfigure[]{
    \includegraphics[height=3.5cm]{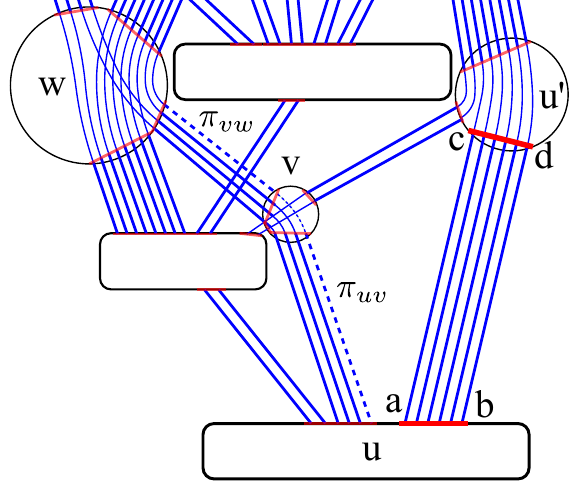}
    \label{fig:pathsdetailed}}
~~~~~~~~~~
\subfigure[]{
	\includegraphics[height=3cm]{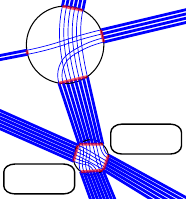}
    \label{fig:bentrect}}
~~~~~~~~~~
\subfigure[]{
	\includegraphics[height=2.5cm]{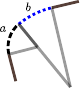}
    \label{fig:biarc}}
~~~~~~~~~~
\subfigure[]{
	\includegraphics[height=2.5cm]{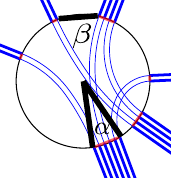}
    \label{fig:radii}}
    \caption{
    (a)~$[ab]$ and $[cd]$ are the bundle bases of edge $uu'$.
    The circles represent the hubs of the intermediate nodes, and the
    rounded rectangles represent the hubs of the obstacle centers.
    Segments $\pi_{uv}$, the hub segment, and $\pi_{vw}$ for the dashed path.
    (b)~An example of bundle bases when not enough free space is found.
    (c)~An example of a biarc consisting of two arcs $a$ and $b$.
    (d)~The hub radius should be large enough that
    $\alpha \le \frac{\pi}{2}$ and $\beta \ge \textit{path separation}$.}
\end{figure}

We surround each node of $\gv$ by a simple closed curve. In the case of an obstacle center it is
the boundary curve of the node, and in the case of an intermediate node it is a circle disjoint from the
obstacles.  We call these curves \df{hubs}, and denote the hub of node $v$ by $h_v$.
For each node $u \in W$ and each edge $uv \in U$ we construct a line segment with
both its ends belonging to $h_v$.  We call these segments the \df{bundle bases} of edge $uv$.
 When only one path passes through an edge its bundle bases collapse to points.

Each path $\pi$ from the bundle of edge $uv$ is represented by a straight
line segment connecting the bundle bases corresponding to $uv$. We denote such a segment
$\pi_{uv}$ and refer it as a \df{bundle segment}. The bundle segments of different paths
are drawn in a particular order, as discussed in Section~\ref{sec:order} below.
 If path $\pi$ passes through consecutive nodes $u,v,w$, then
 we connect the segments $\pi_{uv}$ and $\pi_{vw}$ by
a smooth curve called a \df{hub segment}. We construct the hub segment in a way that it is tangent
to $\pi_{uv}$ and $\pi_{vw}$, and is contained in $h_v$. For the hub segments we use biarcs, following~\cite{biarc}.
A biarc is a smooth curve formed by two circular arcs with a common point, see Fig.~\ref{fig:biarc}.
Therefore, in the final drawing each path is represented as a sequence of straight line
bundle segments and biarc hub segments.

 Next we give a heuristic that moves intermediate nodes of $\gv$ to
 find space for hubs and bundles. Ideally, we would like to keep
the bundle bases large enough to accommodate all the bundle segments with their widths and
separations. However, this will not always succeed, as in Fig.~\ref{fig:bentrect}.

\subsection{Optimization Goals}
We change positions of intermediate nodes of $\gv$ to provide enough space for
drawing paths. For this we compute how much free space we need around bundles
and hubs. The ideal bundle width for edge $e$, denoted by $e_w$, is the sum of
all path widths of the bundle, together with the total of the separations between them. For
example, for a bundle with three paths of widths $1$, $2$ and $2$, and a path
separation of $1$, the bundle width is $7$.

The radius of a hub for an intermediate node $u$ is the minimum of two radii: a
\df{desired} one, and an \df{allowed} one. The desired radius is at
least $\mu e_w$, where $e \in U$ is an adjacent to $u$ edge, and $\mu$ is a
positive constant (see Fig.~\ref{fig:radii}). In our implementation, $\mu=\frac{1}{\sqrt{2}}$.
In addition, the desired radius is chosen large enough that any two bundles are
separated before entering the hub by at least the given edge separation.
We also bound the desired radius from above by some predefined constant, to avoid
huge hubs. The allowed radius should be small enough that: (a)~the hub interior
does not intersect other hubs, and (b)~the hub is disjoint from the obstacles.

We say that the hub of an intermediate node $u$ is \df{valid} if conditions (a) and (b) above
are satisfied, and in addition for each edge $uv$ the straight line segment from the center of
$h_u$ to the center of $h_v$ does not intersect any obstacle (if $v$ is an intermediate
node) or intersects only obstacle $O_v$ (otherwise). The routing graph is \df{valid}
if all intermediate hubs are valid.

Our optimization procedure finds a valid graph $\gv$ while pursuing two objectives. First,
the radii of all hubs should be as close to the desired radii as possible. Second,
we try to keep $\Cost$ of $\gv$ small.

\subsection{Routing Graph Optimization}
Let us denote by $p_v$ the position of node $v$ of $\gv$.
Before the optimization step we have $\gv$ in a valid state. We find hub radii (perhaps very small)
keeping $\gv$ valid, but we might have not enough
space for the hubs and the bundles.
In particular, we might have hubs for which the desired radius is larger than the allowed radius.
The optimization starts by trying to move each such hub away from obstacles into another
valid position, where the allowed radius is closer to the desired radius than before.
This is done in a loop, for which we initially set $r$ to the desired radius of node $v$.
A move is attempted along vector $\ell$ with the direction $\sum_s u_s$, where
$u_s=\frac{\overrightarrow {O_sp_v}}{\overrightarrow {|O_sp_v|}}$, and the sum is taken over all
obstacles $s$ intersecting the circle of radius $r$ with $p_v$ as the center.  The length of $\ell$
is $\theta r$ for some fixed $\theta >1$, which is $1.1$ in our implementation.  If the new position of $v$ is invalid,
we diminish $r$ and repeat the procedure. We stop iterations for a node after $10$ unsuccessful moves.

After this step we obtain more space for the paths, but $\Cost$ is typically increased.
Therefore we again iteratively adjust the positions of intermediate nodes to diminish the cost.
In one iteration we consider an intermediate node $u$ of $\gv$ and try to find a valid position
$p_u$ which maximally decreases $\Cost$. We assume that our optimizations do not change
the capacity segments crossed by each path; hence, we only try to find the local minimum of
ink and path lengths components of the cost. Consider the function which shows the contribution of node $u$ to $\Cost$:
$f(p_u) = \inkcoeff \sum_v |p_u p_v| + \lencoeff \sum_{\pi_{st}} (|p_u p_v| + |p_u p_w|)/|st|$,
where the first sum is taken over all neighbors of $u$, and the second sum is taken over
all paths passing through nodes $v,u,w$. To minimize $f(p_u)$, we use gradient descent method
by moving $u$ along the direction
$D = \inkcoeff \sum_v{ \frac{\overrightarrow {p_up_v}}{|p_up_v|}} +
\lencoeff \sum_{\pi_{st}} \frac{1}{|st|} (\frac{\overrightarrow{p_u p_v}}{|p_u p_v|} + \frac{\overrightarrow{p_u p_w}}{|p_u p_w|})$ using a small fixed step size.
If $\Cost$ is reduced, and the new position is still valid, we repeat the procedure,
otherwise the move is discarded and we proceed with the next node.
We pass a predefined number of times over all nodes of the graph; this number is $2$ in our implementation.

After this optimization procedure we also try to modify the structure of graph $\gv$, if it is beneficial.
We try the following modifications: (1)~shortcut an intermediate node of degree two;
(2)~glue adjacent intermediate nodes together.
These transformations are applied only if they reduce $\Cost$ while keeping the graph valid.

What is the complexity of the above procedure? Let $c$ be the time required to find
out if a circle, or a straight line segment, intersects an obstacle. Using an R-tree~\cite{guttman84}
on the obstacles, one can find out if a circle or a rectangle intersects the
obstacles in $O(c\log |W|)$ time. The number of edges in $\gv$ is $|U|$,
therefore, an iteration optimizing the position of every node of $\gv$ can be
done in $O(c |U|\log |W|)$ time. Since we apply a constant number of iterations,
this is also a bound for the whole procedure.

\section{Ordering Paths}
\label{sec:order}
At this point the routing is completed and the bundles have
been defined. We draw the paths of a given bundle parallel to
the corresponding edge, therefore two paths may need to cross
at a node as shown in Fig.~\ref{fig:ordering}. The order of paths
in bundles affects crossings of paths.  Some path crossings cannot be avoided since they are induced by the edge crossings of $\gv$, while others depend on the relative
order of paths along their common edges, and thus might be avoided.
We would like to find the orders such that only unavoidable crossings remain. Let $P$ be the set of simple paths in $\gv$ computed by path
routing. We address the following problem.

\begin{figure}[ht]
    \center
	\subfigure[]{
	\includegraphics[height=3cm]{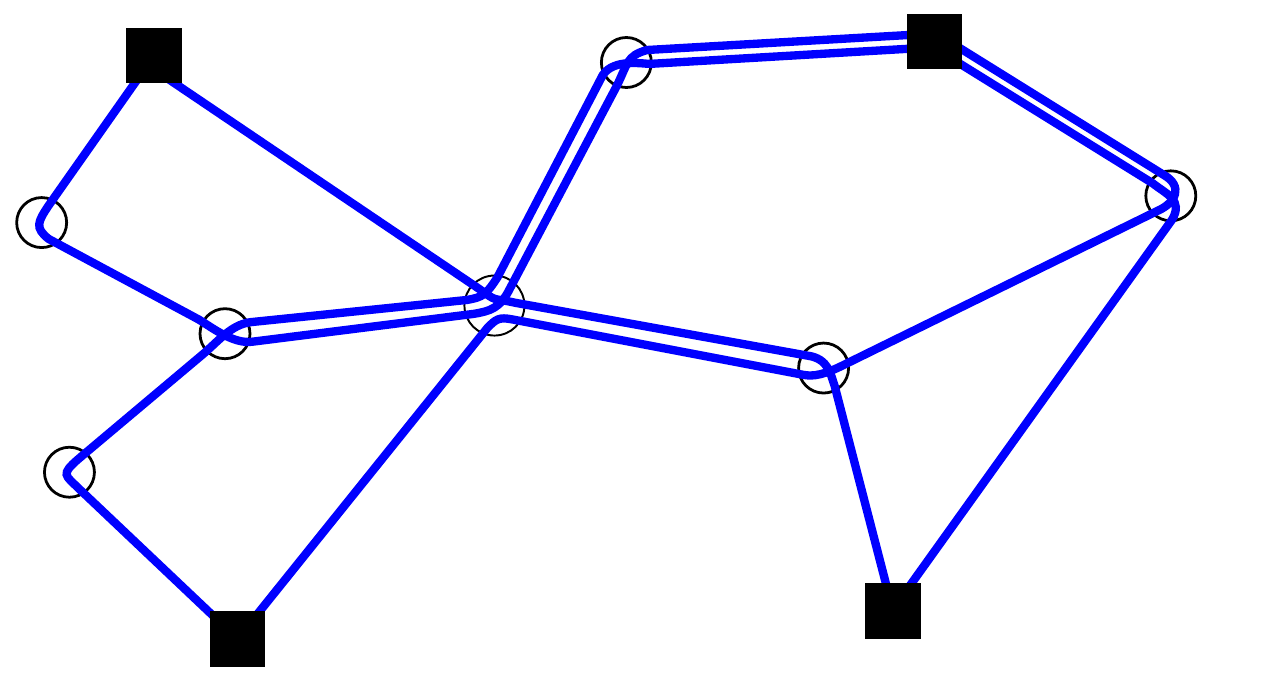}}
~~~~~~~~~~~~~~~~
	\subfigure[]{
	\includegraphics[height=3cm]{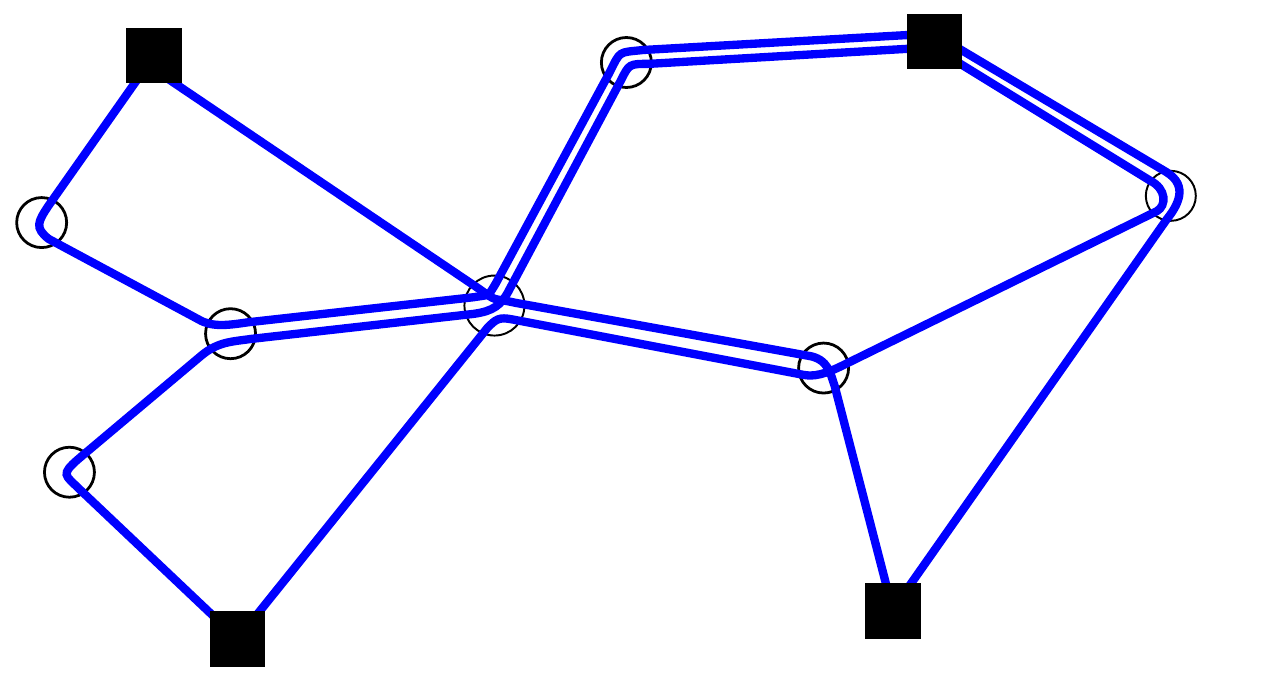}}
    \caption{Graph $\gv$ with 4 terminal (squares) and 7 intermediate (circles) nodes.
    (a)~Ordering with 6 crossings.
    (b)~Optimal ordering with 3 crossings.}
    \label{fig:ordering}
\end{figure}

\begin{problem}
\label{problem:MLCM}
Given the embedded graph $\gv$ and a set of simple
paths~$P$, find an ordering of paths for each edge of $\gv$
that minimizes the total number of crossings among all pairs of paths.
\end{problem}

The complexity class of Problem~\ref{problem:MLCM} is
unknown. However, in our setting paths satisfy an additional property
that makes possible to solve the ordering problem efficiently.
We now describe this property.

For path $(v_1,\dots,v_k)\in P$ the nodes $v_1$ and $v_k$ are
called \df{terminal}, and the nodes $v_2,\dots,v_{k-1}$ are
called \df{intermediate}. In our setting, the paths terminate
at the nodes of $\gv$ corresponding to the centers of the obstacles.
These are the only nodes that can be terminal for paths. The rest of
the nodes are intermediate. Thus we have:

\emph{Path Terminal Property:} No node of $\gv$ is both a terminal of
some path and an intermediate of some path.

We thus consider the following variant of the crossing
minimization problem.

\begin{problem}[Path Ordering]
\label{problem:MLCM-T} Given an embedded graph $\gv$ and a set of
simple paths $P$ satisfying the path terminal property, compute an
ordering of paths for all edges of $\gv$
so that the number of crossings between pairs of paths is minimized.
\end{problem}

We say that two paths
having a common subpath have an \df{unavoidable crossing} if they cross for every
ordering of paths.  An ordering of paths is \df{consistent} if the only pairs of paths
with unavoidable crossings cross (and these only once).
Clearly, a consistent ordering has the
minimum possible number of crossings. We will prove that a consistent ordering
always exists, and we provide two algorithms to construct one,
solving Problem~\ref{problem:MLCM-T}.
The first algorithm is a modification of a method for wire routing in VLSI
design proposed in~\cite{groeneveld89a}.  We call it the ``simple algorithm''
since it utilizes a simple idea, and it is straightforward to implement.
The second is a new algorithm that has a better computational complexity.

\subsection{Simple Algorithm}
To find an ordering of paths we iterate over the edges of
$\gv$ in any order, and build the order for each edge.
Let $P_{uv}$ be a set of paths passing through
the edge $uv$. We fix a direction of the edge, and construct the
ordering on it by sorting $P_{uv}$.  To define the order
between two paths $\pi_1,\pi_2\in P_{uv}$ we walk along their
common edges starting from $u$ until the paths end or the following cases happen
; If we find an edge which was previously processed by
the algorithm, we reuse the order of $\pi_1$ and $\pi_2$
generated for this edge. If we find a fork node $f$, by which we mean
the end of the common sub-path of the two paths, then the order between $\pi_1$ and
$\pi_2$ is determined according to the positions of the next nodes of
the paths following $f$, the path turning to the left is smaller.  Otherwise, if such a fork node is not found
(which means that the paths end at the same node),
we walk along the common sub-path in the reverse direction starting
from node $v$, and again looking for a previously processed edge or a fork.
The above procedure determines a consistent order for all 
pairs except of the pairs of coincident paths; such paths are ordered arbitrarily on the edge.

This ordering creates unavoidable crossings only. As mentioned earlier, the above algorithm is a minor variant
of one in~\cite{groeneveld89a}. The proof of correctness follows the same
lines as in~\cite{groeneveld89a}. Here we discuss the main differences
between the two variants.
First, for the sake of simplicity we do not modify the graph and the paths
during computations. Second, in our setting the nodes of $\gv$  might have
an arbitrary degree; in particular, several paths might share an endpoint,
and thus the degree of the terminals is not bounded. That is why to compare
two paths we traverse their common subpath in both directions. The running
time of our algorithm is also different from that in~\cite{groeneveld89a}.
Our algorithm iterates over the edges
of $\gv=(W,U)$, which takes $O(|U|)$ time. For each edge it sorts paths
passing through it. A comparison of two paths requires at most $O(|W|)$ steps
(maximal length of a common subpath). Let $|P_e|$ be a number of paths passing through edge
$e\in U$. Then the sorting on the edge takes $O(|W||P_e|\log |P_e|)$ time.
Since $\sum_{e\in U} |P_e| = L$, where $L$ is the total length of paths in $P$,
and $|P_e|\le |P|$, the algorithm runs in $O(|W|L \log |P| + |U|)$ time.
However, this is a pessimistic bound. Normally, the length of a common subpath
of two paths is much smaller than $|W|$, and the number of paths passing
an edge is smaller than $|P|$, that makes the algorithm quite fast in practice.

The above algorithm is straightforward to implement. However, the following slightly
more complicated algorithm has a better asymptotic computational complexity.

\subsection{Linear-Time Algorithm}
\begin{figure}[htb]
    \center
	\includegraphics[height=3.2cm]{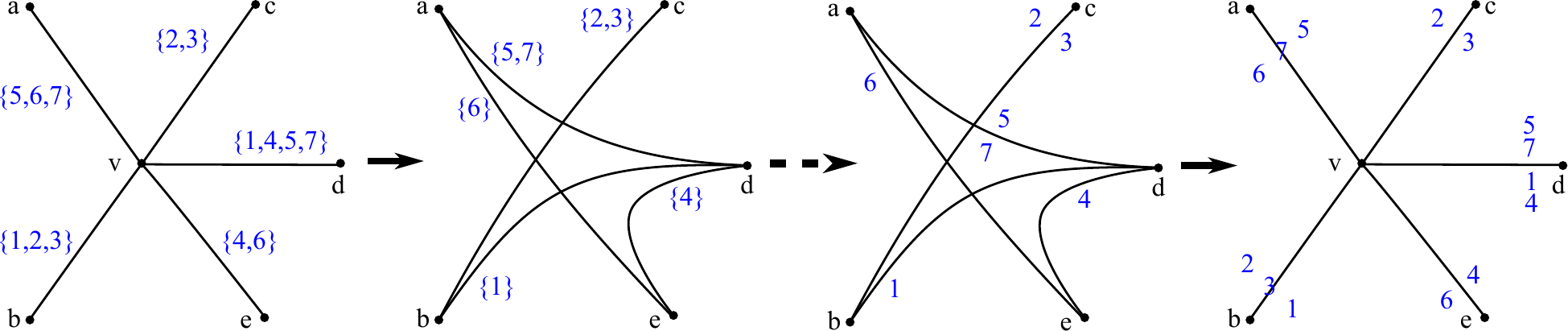}
    \caption{Removal of node $v$; intervening steps; re-insertion of $v$. Paths
    of $P_v$ are numbered $1,\ldots,7$, and are shown in braces when unordered,
    and then with placement indicating their orders.
    The removal of $v$ defines the order between paths 4 and 6 on edge $ve$.
    The order between paths 2 and 3 was calculated in the intervening steps.}
    \label{fig:order_g1}
\end{figure}

The algorithm consists of two phases.  Each step of \emph{Phase~1}
involves the deletion of a node of $\gv$ (Fig.~\ref{fig:order_g1}).
For each non-terminal node $v$ in turn we do the
following. Enumerate the edges incident to $v$ as
$e_1,e_2,\dots,e_t$ in clockwise order, and let $v_1,\dots,v_t$
be the corresponding nodes adjacent to $v$. Represent each path
$\pi\in P_v$ using edges $e_a$ and $e_b$ by an unordered
pair $\{a,b\}$.  For each pair $\{a,b\}$ add a new edge $v_av_b$.
Modify the path represented by $\{a,b\}$ so that it goes directly from
$v_a$ to $v_b$ using this new edge.
The new edges incident to $v_a$ should be inserted into $v_a's$
clockwise order in the position previously occupied by $e_a$,
in the order determined by the indexes of $e_b$. Finally,
delete node $v$ from the graph.

In \emph{Phase~2} the process is reversed.  Initially, all non-terminal nodes have been deleted,
and we re-insert them in the reverse order to which they were deleted, adding
orders of paths to the edges in the process.  Consider the re-insertion
of node $v$. The new order of paths along edge $e_a$ is obtained
by concatenating the orders of paths along the edges $v_av_b$.

We now give some details of the implementation.
First, create a new graph $H$.
Initially, $H$ is the subgraph of $\gv$ induced the union of all the
paths of $P$.  Each path in $P$ is stored as a list of nodes in $H$.
For an edge $e$, let $P_e$ denote the set of paths containing $e$. We
assume that for every node $v$ of $H$, the list of edges
incident to $v$ in clockwise order is given. Note that
these lists are dynamic since $H$ undergoes deletions of nodes.
We keep track of the deletions in a forest $F$. Initially, $F$
consists of isolated nodes corresponding to the edges of $H$. When
a node is deleted and an edge $e$ is replaced by new edges, we
add them in $F$ as children of the node corresponding to $e$.
For instance, $P_{dv}$ contains paths 1,4,5, and 7 in
Fig.~\ref{fig:order_g1}. When node $v$ is processed, set $P_{dv}$
is split into new subsets $P_{de}$, $P_{da}$, and $P_{db}$.
The (clockwise) order of subsets is important. Then we replace
edge $dv$ by edges $de$, $da$, and $db$ in graph $H$, in paths $P_{dv}$,
and in order of edges around $d$. We stress here that
graph $H$ might contain multiple edges after this operation,
since there could be another path passing from $d$ to any of $e$, $a$, or $b$.
It is important to keep multiple edges and corresponding
sets of paths separately. The first phase finishes when all
non-terminals are deleted from $H$.

In the second phase, we process each tree in $F$ in bottom-up
order from children to parents. The list of paths for a node of $F$ is
simply the concatenation of the lists of its children. The leaves of
$F$ correspond to the original paths of $P$.

\begin{theorem}
Given the graph $\gv=(W,U)$, a set of simple paths $P$
satisfying the path terminal property, and a clockwise order of
the edges around each node, the above algorithm computes a consistent ordering of
paths in $O(|W|+|U|+L)$ time, where $L$ is the
total length of paths in $P$.
\end{theorem}

\begin{proof}
\emph{Correctness}.
The ordering of paths computed by our algorithm is consistent
since the deletion of node $v$ adds only unavoidable crossings.
Two paths $\pi_1$ and $\pi_2$ will produce a crossing only when
the last node of their common subpath is deleted and the
clockwise order of the nodes around $v$ is $\cdots v_a\cdots
v_b\cdots v_{a'}\cdots v_{b'}\cdots$, where $\pi_1=\cdots v_a v
v_{a'}\cdots$ and $\pi_2=\cdots v_b v v_{b'}\cdots$.

\emph{Running time}.
The time for processing node $v$ (the deletion of $v$) is
$O(1+d_{v,H}+s_v)$, where $d_{v,H}$ is the degree of $v$ in
$H$ at the current step and $s_v$ is the number of paths passing
through $v$. The claimed bound follows since $d_{v,H}\le d_{v,\gv}+s_v$.
\end{proof}

We note that the size of the input for Problem~\ref{problem:MLCM-T}
is $\Omega(|W|+|U|+L)$, which is the same as the
running time of our algorithm if the clockwise order
of edges around nodes is part of the input. If this order is not
given in advance, the time complexity of the algorithm is
$O(|W|\log |W|+|U|+L)$. Since the length of a path does not exceed
$|W|$, we have $L= |W||P|$. Therefore, another estimate for the running time is 
$O(|W|\log |W|+|U|+|W||P|)$.

\subsection{Distribution of crossings}
Consistent orderings are not necessarily unique. We noticed that the
choice of a particular one may greatly influence the quality of the final
drawing. The following property might appear
desirable. An ordering of paths is \df{nice} if it is consistent, and for any
two paths, their order along all their common edges is the same
(i.e.\ they may cross only at an endpoint of their common
sub-path). We next analyze the existence of nice orders.

\begin{proposition}
If $\gv$ is a tree, then there is a nice order of paths $P$.
\end{proposition}

\begin{proof}
To build a nice order of paths on a tree, we use the algorithm below.
Initially, choose any terminal node to be the root, and then consider the edges starting at the leaves
towards the root. The first time we encounter a common edge of two
paths will be at the end of their common sub-path. Thus, the crossing
between them (if needed) will take place on that edge. To ensure that the
crossing goes to the endpoint of the sub-path, we direct the
considered edge towards the root of the tree.
\end{proof}

We found an example of $(\gv,P)$ having no nice ordering.

\begin{proposition}
There exists $(\gv, P)$ that admits no nice ordering.
\end{proposition}

\begin{proof}
Consider a gadget with four paths shown in Fig.~\ref{ex1}.
It consists of a black path, two blue paths (called left and right
blue path respectively) and a central red path.
In any nice ordering the first (left) blue path must be
above the black path or the second (right) blue path must be
below the black path. Otherwise, if the left blue path is drawn below and
the right blue path is drawn above the black path, then red and black
paths do not intersect nicely (that is, their order along common
edges is not the same).

\begin{figure}[htb]
\includegraphics[width=\linewidth]{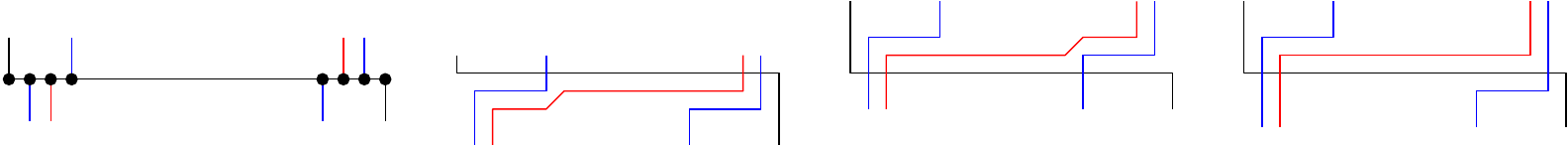}
\centering
\caption{The gadget and 3 ways to draw it.}
\label{ex1}
\end{figure}

Now consider the graph $\gv$ with $7$ gadgets in Fig.~\ref{ex2}, where
only black paths are shown. The gadgets do not share any paths except
(a) gadgets 1 and 4 share the blue path, and (b) gadgets 1 and 7
share the blue path. In a nice ordering either the left blue path for path 1
is above it or the second one is below it. In the first case, the
blue paths of gadget 4 do not satisfy the above property, see
Fig.~\ref{ex3} (a). This is a contradiction. In the second case, the blue
paths of gadget 7 do not satisfy the above property, as shown in Fig.~\ref{ex3} (b).
Again, this is a contradiction, which proves that there is no nice ordering
of considered paths on $\gv$.
\end{proof}

\begin{figure}[t]
	\subfigure[]{
    \label{ex2}
	\includegraphics[width=0.45\linewidth]{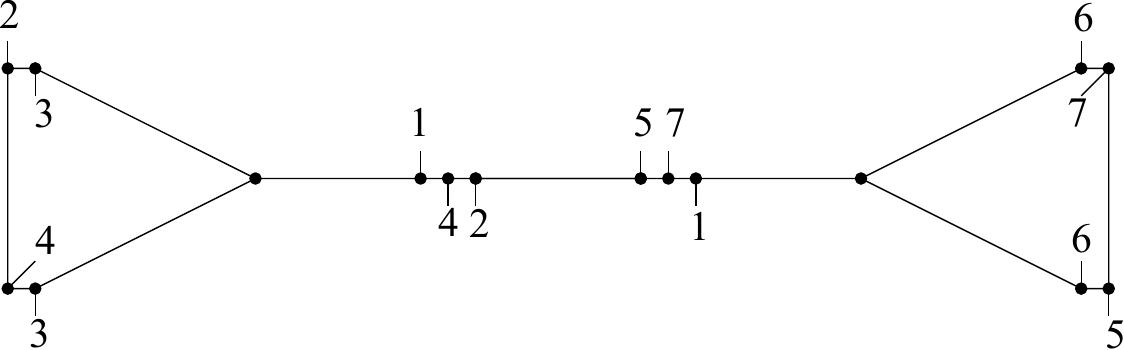}}
~~~~~~~~~~~
	\subfigure[]{
    \label{ex3}
	\includegraphics[width=0.45\linewidth]{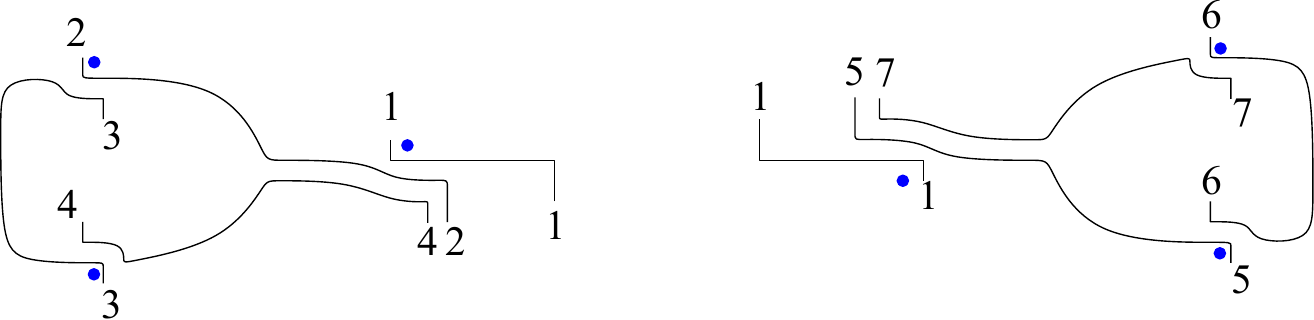}}
	\centering
	\caption{(a)~Graph $\gv$ consisting of 7 gadgets. Each black path starts at a leaf marked
    by a number and ends at a leaf with the same number.
    (b)~Two possible cases; the blue dots mark the location of the gadget's blue paths.}
\end{figure}

\section{Results}
\label{sec:experiments}
We implemented our algorithm in the MSAGL package~\cite{lev07}. Edge bundling was
applied for synthetic graph collections and several real-world graphs
(see~\cite{sergey10} for a detailed description of our dataset).
Unless node coordinates are available, we used the tool to position
the nodes. All our experiments were run on a 3.1~GHz quad-core machine
with 4 GB of RAM. Table~\ref{algorithm_speed} gives measurements of the
method on some test cases.

\begin{table}[t]
    \begin{tabularx}{\linewidth}{|>{\scriptsize}m{0.11\linewidth}|>{\scriptsize}m{0.11\linewidth}|>{\scriptsize}X|>{\scriptsize}X|>{\scriptsize}X|>{\scriptsize}X|>{\scriptsize}X|>{\scriptsize}X|>{\scriptsize}X|>{\scriptsize}X|>{\scriptsize}X|>{\scriptsize}X|}
        \hline
        {Graph}  & {source} & {$|V|$} & {$|E|$} & {$|W|$} & {$|U|$} & {Routing} & {Optimization} & {Ordering} & {Overall} \\
        \hline
        \texttt{tail}         	       & \cite{sergey10} & 21   &   68  &  105 &  348    & 0.13 & 0.17  & 0.01 & 0.34 \\
        \texttt{notepad}          	   & \cite{sergey10} & 22   &  113  &  198 &  776    & 0.03 & 0.29  & 0.01 & 0.35 \\
        \texttt{airlines}        	   & \cite{cui08}    & 235  & 1297  & 1175 &  5297   & 0.47 & 2.67  & 0.15 & 3.32 \\
        \texttt{jazz}           	   & \cite{gephi}    & 191  & 2732  &  955 &  4478   & 0.56 & 3.25  & 0.18 & 4.04 \\
        \texttt{protein}        	   & \cite{gephi}    & 1458 & 1948  & 7290 &  32585  & 1.45 & 11.92 & 0.10 & 13.52 \\
        \texttt{power grid}     	   & \cite{gephi}    & 4941 & 6594  & 24705&  109779 & 7.97 & 18.01 & 0.18 & 26.31 \\
        \texttt{Java}         	       & GD'06 Contest   & 1538 & 7817  & 7690 &  32712  & 4.08 & 30.17 & 0.96 & 35.35 \\
        \texttt{migrations}        	   & \cite{cui08}    & 1715 & 6529  & 8575 &  41451  & 3.89 & 30.57 & 1.16 & 35.75 \\
        \hline
    \end{tabularx}
    \centering
    \smallskip
    \caption{Statistics on test graphs and performance of the algorithm steps (in seconds).}
    \label{algorithm_speed}
\end{table}

\paragraph{Performance}
The table shows the CPU times of the main algorithm steps. As can be seen, ordered
bundles can be constructed for graphs with several thousand of nodes and
edges in less than a minute. For medium sized graphs with hundreds of
nodes the algorithm can be applied in an interactive environment. The
most time expensive step is \emph{Routing Graph Optimization}.
The explanation for its long running time is twofold. First, it
continues to modify $\gv$ as long as it reduces $\Cost$, and we do not
have a good upper bound for the number of iterations.  Second, the step
makes a lot of geometric operations such as checking an
intersection between circles, polygons and line segments. These operations
are theoretically constant time but computationally expensive
in practice.  We conclude that the optimization of this step is
the primary direction for the running time improvement of our algorithm.

\paragraph{Selection of parameters}
We describe the influence of the parameters in $\Cost$ on the final drawing.
The weight of the ink component $\inkcoeff$ can be considered as the
``bundling strength'', since larger values of $\inkcoeff$ encourage more bundling. The path length component
is complimentary and has the opposite effect.
The capacity component plays an important role for graphs with large and
closely located node labels (see for example the narrow channel between nodes $S2$ and $S5$ in Fig.~\ref{fig:parameters2}). Based on our experiments, we
recommend selecting $\lencoeff$ significantly larger than $\inkcoeff$,
otherwise some paths can be too long. The capacity weight $\capcoeff$
is set larger by order of magnitude than the other coefficients.
In our default settings $\inkcoeff = 1$, $\lencoeff = 500$,
and $\capcoeff = 10(\inkcoeff + \lencoeff)$.

We may also adjust the widths of individual paths and path separation
(see Fig.~\ref{fig:parameters}). In the extreme case when the path separation and path widths
are set to zero, we obtain a drawing with overlapped edges.

\paragraph{Examples}
We now demonstrate ordered edge bundling algorithm on some real-world examples.
A \textit{migration} graph used for comparison of edge bundling algorithms is shown
in Fig.~\ref{fig:migration}. In our opinion, on a global scale ordered edge
bundles are aesthetically as pleasant as other drawings of the graph
(see e.g. \cite{cui08,hu11,holten09,lambert10}).
On a local scale, our method outperforms previous approaches by arranging edge intersections.
A smaller example of edge bundling is given in Fig.~\ref{fig:tail}. It shows
another advantage of our routing scheme. Multiple edges are
visualized separately making them easier to follow (compare the edge
between nodes \textit{Editor} and \textit{Application} on the original and bundled drawings).

\begin{figure}[t]
	\subfigure[]{
	\includegraphics[width=0.50\linewidth]{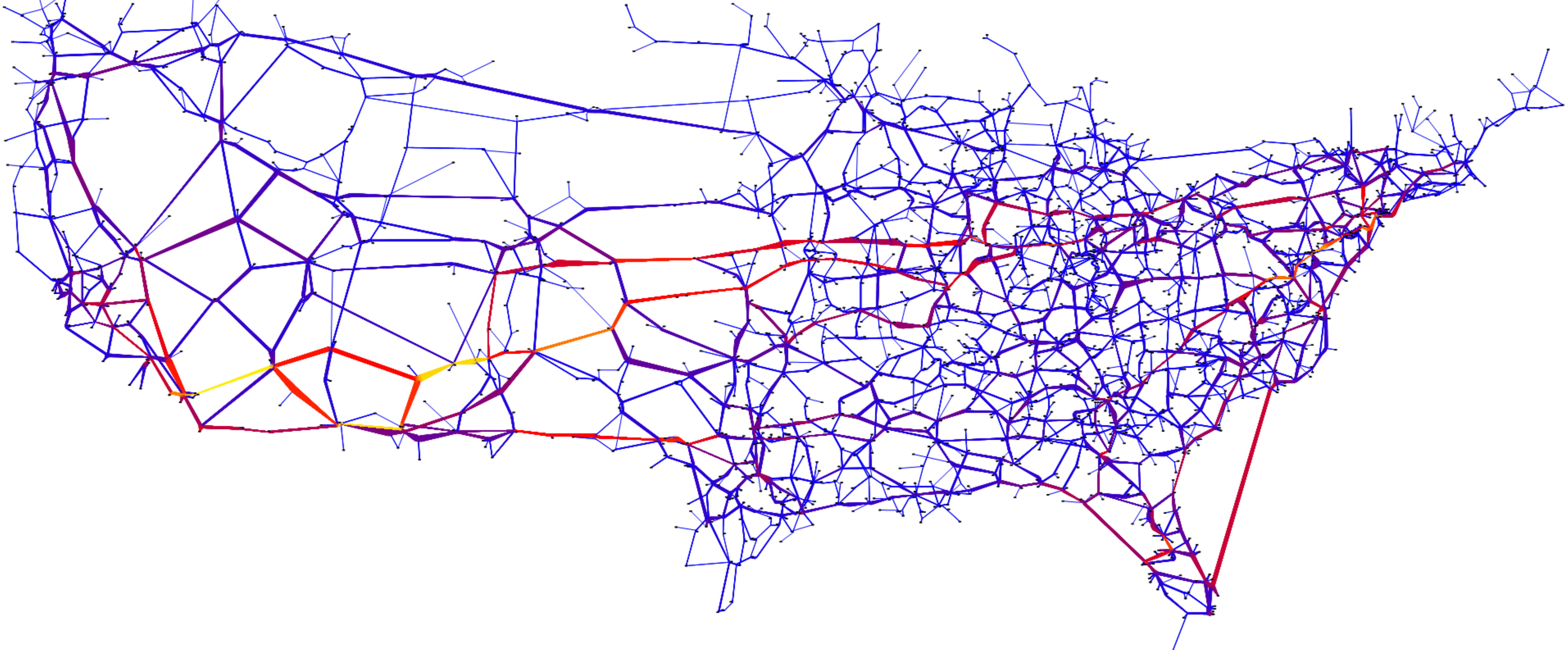}}
~~~~~~
	\subfigure[]{
	\includegraphics[width=0.4\linewidth]{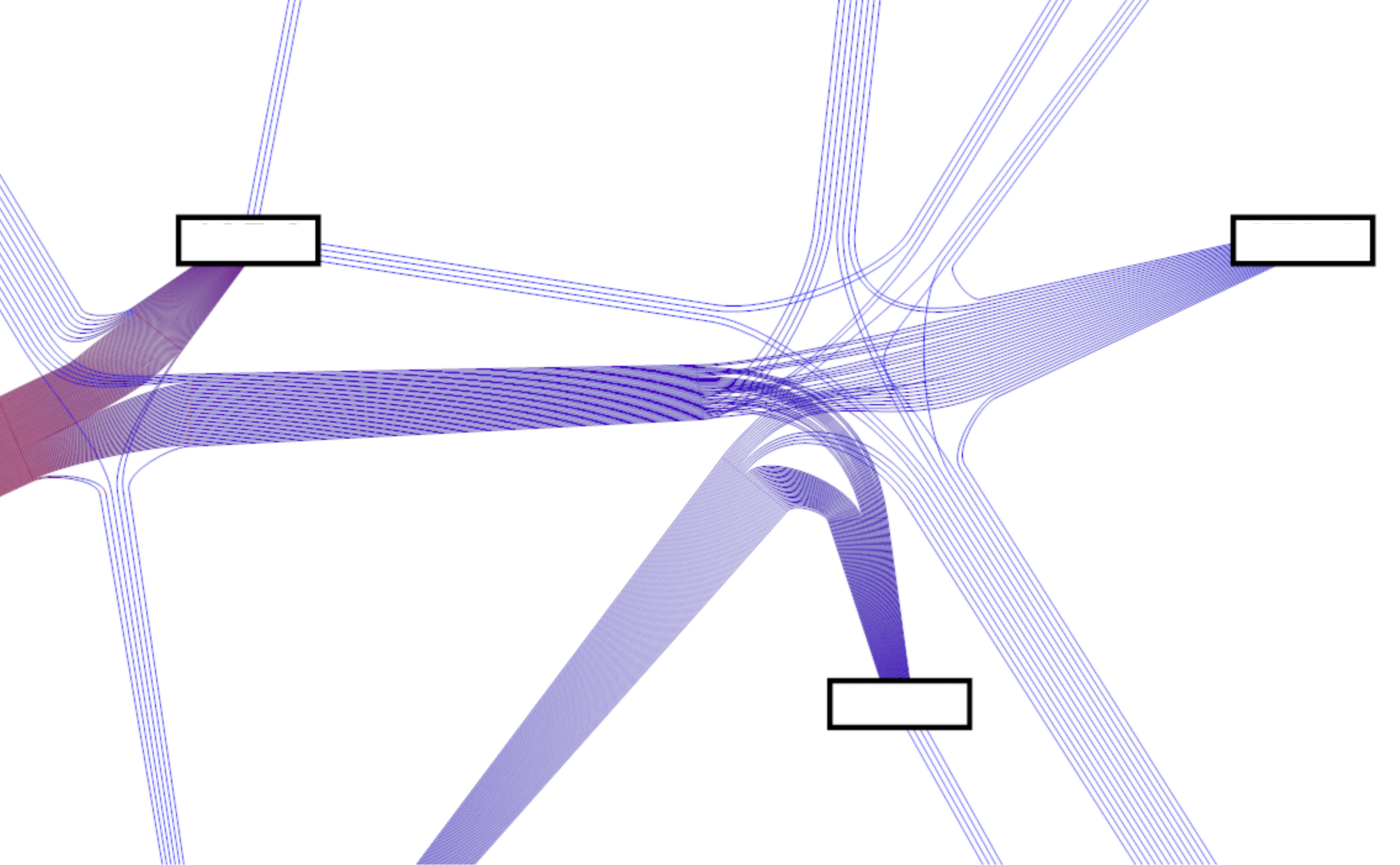}}
	\centering
	\caption{\texttt{Migration} graph. (a) Overview. (b) Detail.}
    \label{fig:migration}
\end{figure}

\begin{figure}[t]
	\subfigure[]{
	\includegraphics[height=6cm]{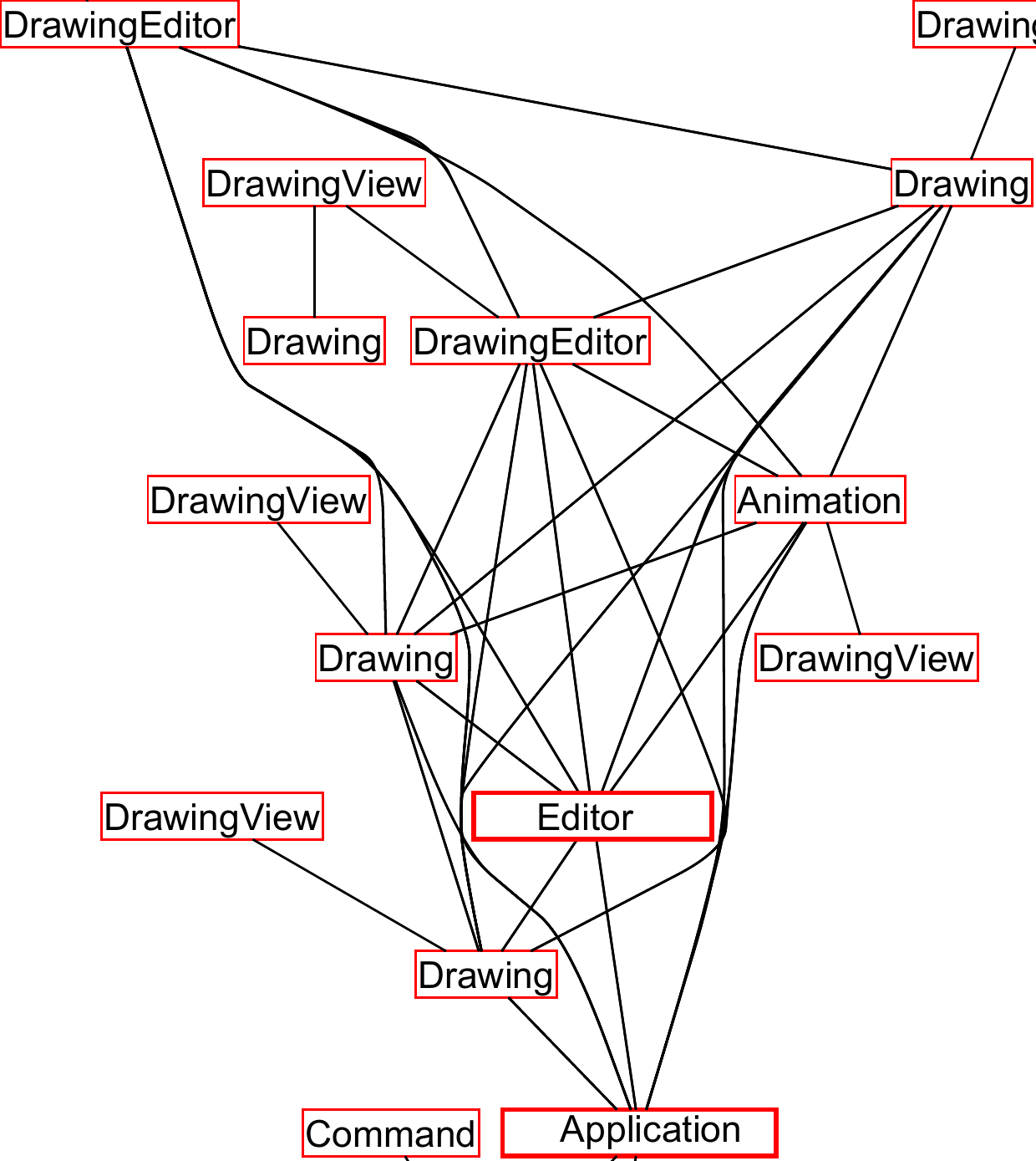}}
~~~~~~~~~~
	\subfigure[]{
	\includegraphics[height=6cm]{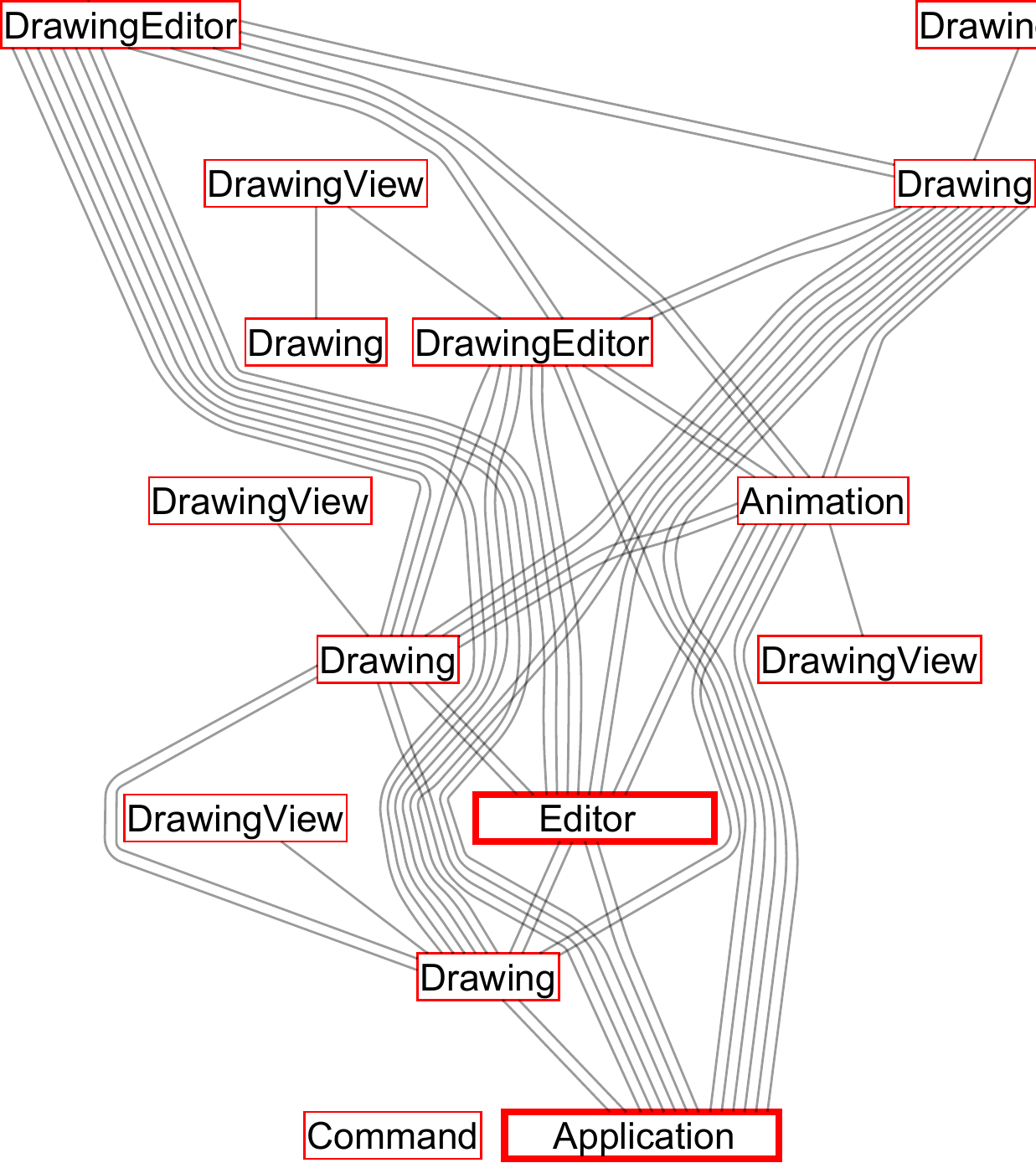}}
	\centering
	\caption{\texttt{Tail} graph. (a) Standard routing computed with the visibility graph approach.
    (b) Edge bundling routing.}
    \label{fig:tail}
\end{figure}

\begin{figure}[t]
	\subfigure[]{
    \label{fig:parameters1}
	\includegraphics[height=6cm]{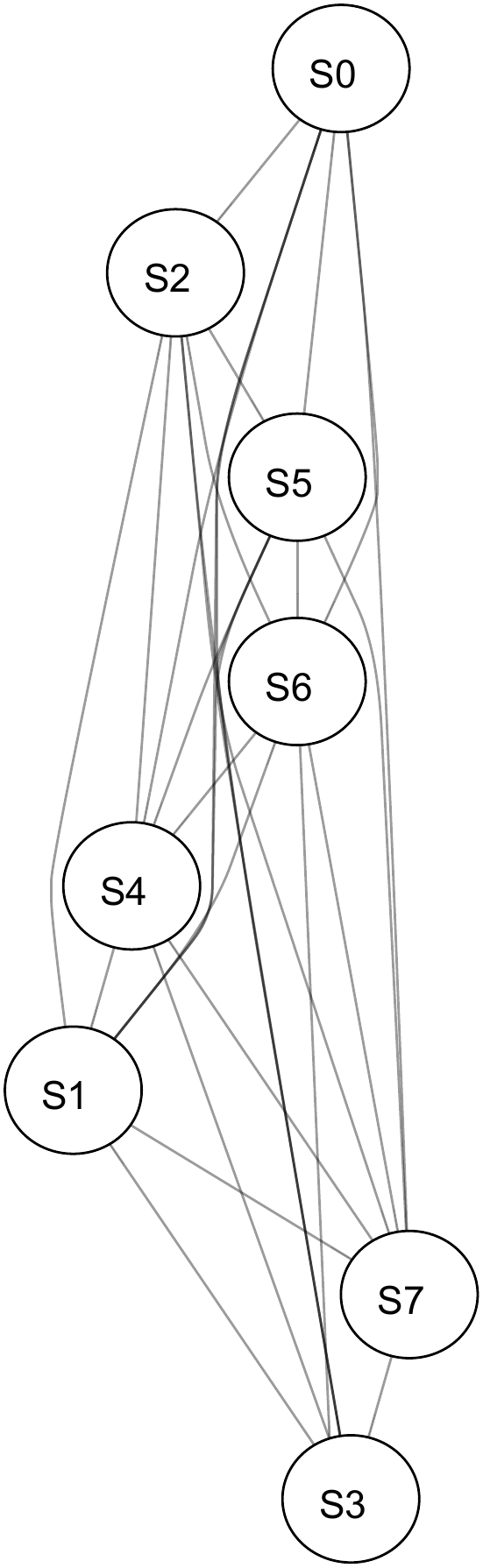}}
~~~~~~~~~~
	\subfigure[]{
    \label{fig:parameters2}
	\includegraphics[height=6cm]{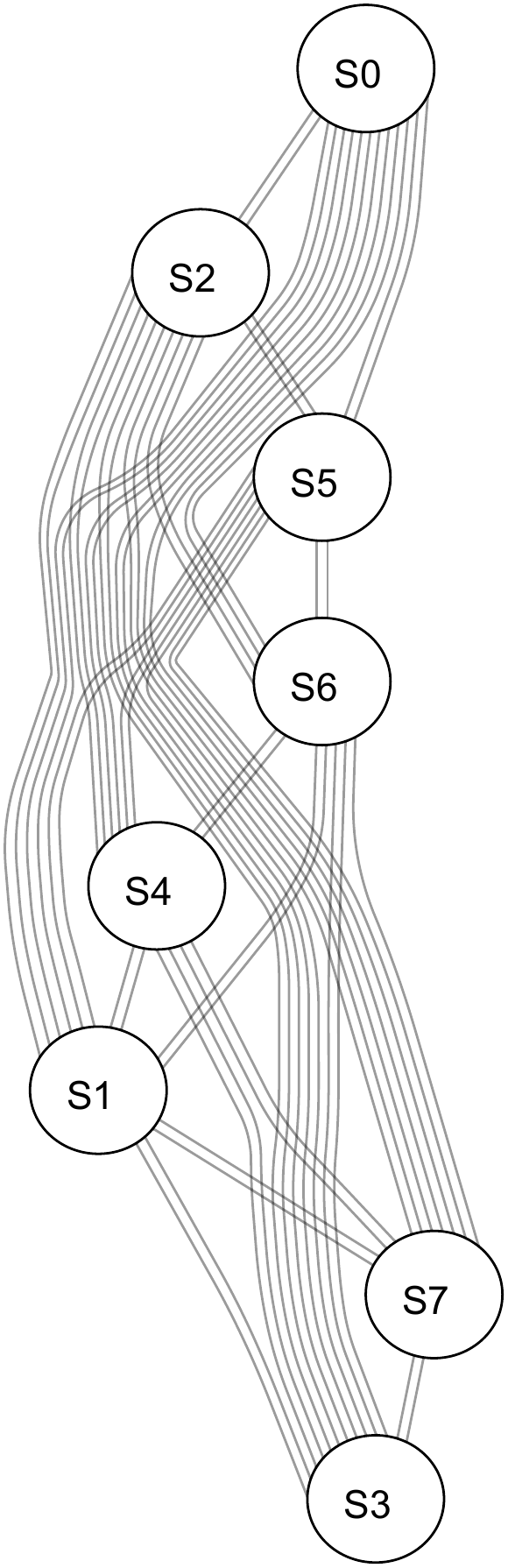}}
~~~~~~~~~~
	\subfigure[]{
    \label{fig:parameters3}
	\includegraphics[height=6cm]{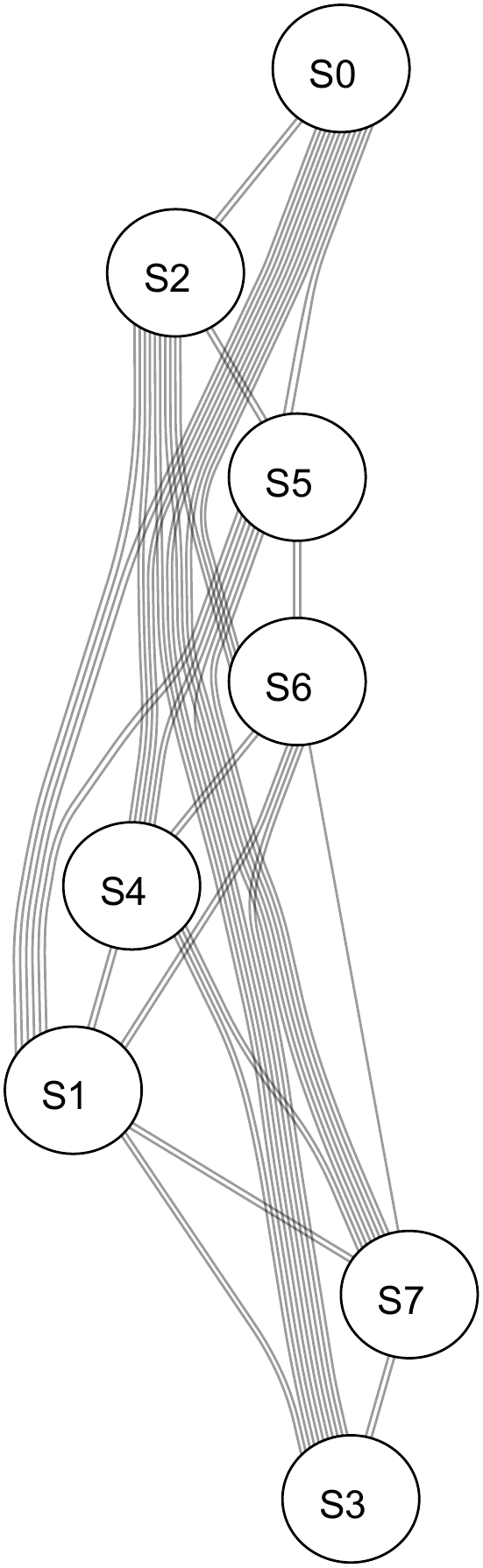}}
~~~~~~~~~~
	\subfigure[]{
    \label{fig:parameters4}
	\includegraphics[height=6cm]{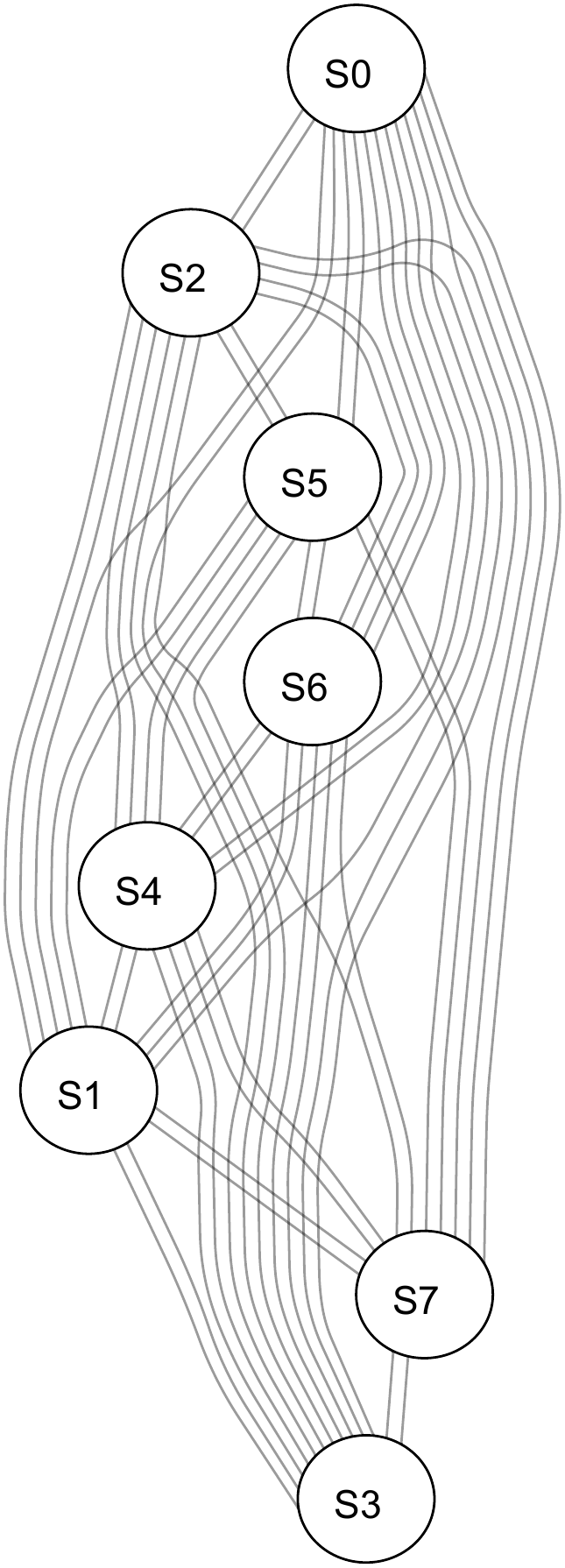}}
~~~~~~~~~~
	\subfigure[]{
    \label{fig:parameters5}
	\includegraphics[height=6cm]{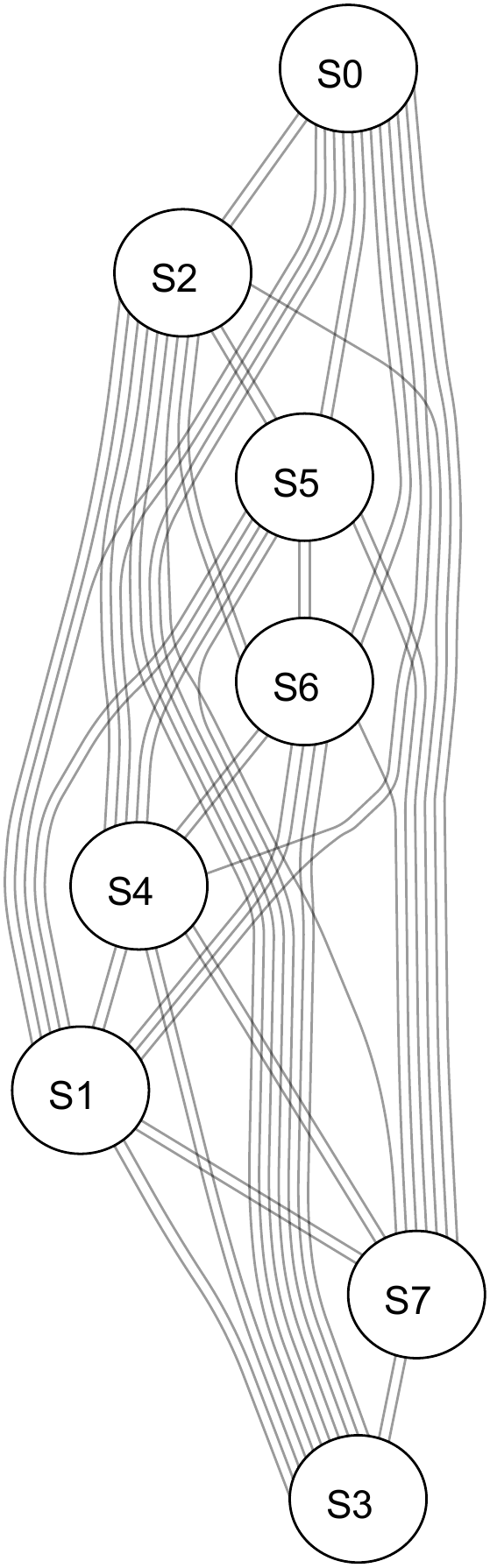}}
	\centering
	\caption{(a)~Classical edge routing produced with visibility graph with many edge overlaps.
    (b)~No capacity overflow penalty, $\capcoeff=0$.
    (c)~Small edge separation.
    (d)~Large edge separation.
    (e)~Edge bundling with the default settings.}
    \label{fig:parameters}
\end{figure}

\paragraph{Limitations}
The main limitation of our method that we are aware of is in
choosing radii and positions for hubs. Sometime we do not utilize the 
available space in the best way. We plan next to improve this heuristic.

\section{Conclusions and Future Work}
\label{sec:conclusion}
We have presented a new
edge routing algorithm based on ordered bundles that improves
the quality of single edge routes when compared to existing
methods. Our technique differs from classical edge bundling, in
that the edges are not allowed to actually overlap, but are run
in parallel channels. The algorithm ensures that the nodes do
not overlap with the bundles and that the resulting edge paths
are relatively short. The resulting layout highlights the edge
routing patterns and shows significant clutter reduction.

An important contribution of the paper is an efficient
algorithm that finds an order of edges inside of bundles with
minimal number of crossings. As mentioned above, this order is
not unique.  The question of choosing the best order is a topic for
future research.

A further possible direction concerns dynamic
issues of edge bundling algorithm. First, a user may want to
interactively change a bundled graph or change node positions.
In that case, a system should not completely rebuild a drawing,
but recalculate affected parts only.  Second, a small deviation
of algorithm parameters (e.g. ink importance $\inkcoeff$) may
theoretically involve a full reconstruction of the routing,
while a smooth transformation is preferable.

\end{document}